\documentclass[12pt,a4paper,oneside]{article}

\usepackage{amsmath,amssymb,mathtools,amsthm,amscd,rotating,array,stmaryrd}
\usepackage[authoryear,round]{natbib}
\usepackage{enumitem,multirow}
\usepackage{psfrag,epsfig,boxedminipage}
\usepackage{pgfplots}
\usepackage{dsfont}
\pgfplotsset{compat=1.11}

\usepackage[onehalfspacing]{setspace}
\usepackage[utf8]{inputenc}
\usepackage[T1]{fontenc}
\usepackage[paper=a4paper,left=30mm,right=30mm,top=40mm,bottom=40mm]{geometry}
\usepackage{mathrsfs}
\usepackage{bbm}

\usepackage{multirow}

\usepackage{tikz}
\usepackage{dsfont}

\usepackage{caption}
\newcommand{\IN}{\ensuremath{\mathbb{N}}}
\newcommand{\IR}{\ensuremath{\mathbb{R}}}

\newcommand{\IQ}{\ensuremath{\mathbb{Q}}}
\newcommand{\IE}{\ensuremath{\mathbb{E}}}
\newcommand{\IS}{\ensuremath{\mathbb{S}}}

\newcommand{\I}{\mathds{1}}

\newcommand{\bd}[1]{\ensuremath{\boldsymbol{#1}}}

\graphicspath{{img/}}

\newtheoremstyle{mystyle}
{9pt}
{9pt}
{\itshape}
{}
{\bfseries}
{}
{\newline}
{}

\newtheoremstyle{mystyle2}
{9pt}
{9pt}
{}
{}
{\bfseries}
{}
{\newline}
{}

\newtheoremstyle{mystyle3}
{9pt}
{9pt}
{}
{}
{\bfseries\itshape}
{}
{\newline}
{}

\theoremstyle{mystyle}

\newtheorem{Sa}{Theorem}[section]
\newtheorem{Lem}{Lemma}[section]

\newtheorem{Kor}{Corollary}[section]

\theoremstyle{mystyle2}

\theoremstyle{mystyle3}

\usetikzlibrary{patterns}

\begin{document}

\title{A consistent nonparametric test of the effect of dementia duration on mortality}
\author{L Radloff, R Wei{\ss}bach, C Reinke, G Doblhammer \\[2mm] \textit{\footnotesize{Chairs of Statistics/Econometrics and Empirical Social Sciences/Demography,}}   \\[-2mm]
        \textit{\footnotesize{Faculty for Economic and Social Sciences,}} \\[-2mm]
        \textit{\footnotesize{University of Rostock \& DZNE}} \\
        }
\date{ }
\maketitle

\renewcommand{\baselinestretch}{1.5}\normalsize

\begin{abstract}
A continuous-time multi-state history is semi-Markovian, if an intensity to migrate from one state into another, depends on the duration in the first state. 
Such duration can be formalised as covariate, entering the intensity process of the transition counts. We derive the integrated intensity process, prove its predictability and the martingale property of the residual. In particular, we verify the usual conditions for the respective filtration. As a consequence, according to Nielsen and Linton (1995), a kernel estimator of the transition intensity, including the duration dependence, converges point-wise at a slow rate, compared to the Markovian kernel estimator, i.e when ignoring dependence. By using the rate discrepancy, we follow Gozalo (1993) and show that the (properly scaled) maximal difference of the two kernel estimators on a random grid of points is asymptotically $\chi^2_1$-distributed. As a data example, for a sample of 130,000 German women observed over a period of nine years, we model the mortality after dementia onset, potentially dependent on the disease duration. As usual, the models under both hypotheses need to be enlarged to allow for independent right-censoring. We find a significant effect of dementia duration, nearly independent of the bandwidth. 
 \\[2mm]
\end{abstract}

%\newpage
\section{Introduction} \label{intro}
Continuous-time multi-state Markovian histories \cite[see e.g.][]{Hougaard.1999,Andersen.2002,kimjame2011} can be generalized to semi-Markov models \cite[see e.g.][]{limnios2001} to allow for dependence of the transition intensities on state duration.  In the analysis of morbidities, the duration since one health event might influence the risk of another, even adjusted for age. For the relation between stroke and dementia, \cite{Pendlebury.2009}, as well as \cite{Corraini.2017}, find descriptive evidence of differences in dementia hazard for different elapses of stroke experiences. As examples from business operations research, \cite{Lando.2002} perform a Cox-test to draw conclusion for different transitions between rating classes, dependent on the duration within a class. \cite{Koopman.2008} even fit a parametric semi-Markov model to rating histories. 

In general, event histories are usually studied by means of counting processes \cite[for semi-Markov models, see][]{keidinjens} and for its asymptotic behaviour, one routinely decomposes them into integrated intensity process and a martingale. Among the `usual regularity conditions' for the filtration, the right-continuity requires effort. 
More specifically,  the duration in the current state as a covariate allows estimating the intensity of a progressive healthy-ill-dead history using a nonparametric regression.  
Smoothing the integrated intensity process, as in a Nadaraya-Watson regression, the estimator is asymptotically normal, when the assumptions of \citet[][Theorem 1]{Linton.1995} hold true. The main assumption to be proven is the martingale property of the process, which results from subtracting the rather obvious compensator from the counting process. This can be achieved simply by exploiting elementary properties of the conditional expectation. In order to make the model useful for our
event-history data, we extend all results to right-censored histories. We also cover the general case of a semi-Markov process, by which not only chronic diseases like dementia can be treated, but also `jumps back', i.e. non-chronic diseases like infections. We then adopt
a powerful test from \cite{Gozalo.1993} for the omission of the state duration. The test compares the  kernel estimators with and without Markov assumption, scaled by the (estimated) standard errors of the unrestricted estimator. The maximal difference of the two, on a random grid of points, is penalized to ensure that the maximum is unique. The test statistic, represented as a martingale transform, is asymptotically $\chi^2_1$-distributed. As prerequisites, we need to apply the Cramer-Wold device to derive a multivariate convergence of the
transition intensity estimator, and also need to show weak uniform consistency of the standard error's estimator. As the data example, we study whether the death intensity depends on the time since a preceding dementia onset, or is only elevated by the dementia disease itself. From a simple sample of 130.000 files from a German health insurance company, we find, indeed, that the disease duration has an effect on the mortality forecast for women.

%\newpage

\section{Semi-Markov Model and Hypothesis}

\subsection{Characterisation and statistics} \label{sect21}

For the history $\mathcal{X}$, let denote by $\IS$ the finite set of possible states, and by $(\Omega, \mathfrak F, P)$ the probability space. From a sequence of $P$-a.s. positive and not necessarily independent `inter-arrival times' $T_1, T_2,...$, define the time of the $m^{th}$ transition $Z_m := \sum_{j=1}^mT_j, m\in \IN_0$ (especially $Z_0 = 0$). Let $S_0,S_1,...,$ be a sequence of $\IS$-valued random variables with $S_j \ne S_{j+1}, j\in\IN_0$. Now we define the process $\mathcal{X}$ by $\mathcal{X}(t) := \sum_{m=0}^\infty \I_{[Z_m,Z_{m+1})}(t) S_m$ for $t>0$ and refer to it as semi-Markovian if
\begin{equation} \label{semmarkprop}
	\begin{split}
&		P(S_{m+1}  =s, T_{m+1}\le x| (S_0, Z_0),...,(S_m, Z_m) = (r, z)) \\
		& = P(S_{m+1}=s, T_{m+1}\le x| S_m = r,Z_m = z) \\
		& = P(S_{2}=s, T_{2}\le x| S_1 = r,T_1 = z)=: F_z^{rs}(x).
	\end{split}
\end{equation}
Similar to a Markov process, transition probabilities do not depend on the number of past transitions, $m$, and the future depends on the current state $S_m$. The generalisation is that, additionally, the future may depend on the time since the last transition, $T_m$. 

We consider a heterogeneous process in which the transition intensities vary in the (deterministic) time, namely age, \cite[see e.g.][Chapt. 3]{limnios2001}. Conditional on the past, the distribution of $T_m$ is described by its survival function $S_z^r(x):= 1 - \sum_{s\in\IS} F_z^{rs} (x)$. We assume differentiability, denote $f_z^{rs} := (F_z^{rs})'$ and assume: 
\begin{itemize}
	\item[(L)] $f_z^{rs}(x) / S_z^r(x)$ are uniformly bounded across $r,s\in\IS$ and $z, x \in [0,1].$
\end{itemize}
Note that, by (L), $f_z^{rs}(x)$ is also uniformly bounded because of $S_z^r(x)<1$. Let us further denote $\mu(t):=\sum_{m=0}^\infty m \I_{[Z_m, Z_{m+1})}(t) = \sum_{m=1}^\infty \I_{\{Z_m\ge t\}}$ the number of jumps of $\mathcal{X}$ until $t$. Being in $t$, it also indicates the index $m$, to which the state of $\mathcal X$, $S_m$, belongs. How often $\mu$ jumps will later help in finding bounds, and the following is the result of a short calculation.
\begin{Lem} \label{dMu.lem}
There is a constant $C$, independent of $k$, $r$, $d$ and $t-d$,  such that for sufficiently small $h$,
$P(\mu(t+h) - \mu(t) \ge k | S_m = r, Z_m = t-d, T_{m+1} > d ) \le C^k h^k$.
\end{Lem}

In order to define the intensity of a transition from state $r$ to state $s$, the duration of $\mathcal{X}$ in its current state, $\tilde D(t):= \sum_{m=0}^\infty (t-Z_m) \I_{[Z_m, Z_{m+1})}(t)$, is important. It is (see Appendix \ref{furthres} for proof):
\begin{equation} \label{alpha_def.leam}	
\lim_{h\rightarrow0} \frac 1h P(\mathcal{X}(t+h) = s| \mathcal{X}(t) = r, \tilde D(t) = d, \mathcal{X}(u), u\le t)
=\frac{f_{t-d}^{rs}(d)}{S_{t-d}^r(d)} =:\alpha_{rs}(t, d)
\end{equation}

When formulating the test later on, we will estimate $\alpha_{rs}(t, d)$. We must select a support and will use a compact set, without loss of generality $t,d\in[0,1]$.
We are only interested in one particular pair $\{r\ne s\}$, so that we suppress the index in the notation from now on. Also, it is easier to illustrate the proofs by specializing in a model with only three progressive states. (It will also suffice for our data example in Section \ref{app}.) The generalisation to the semi-Markov model will either have a proof in the Appendix (e.g. Appendix \ref{semimarkmart} for Theorem \ref{martingale_property.th}), or arguments will at least be outlined, most of the time. 
This simplifies $(\mathcal{X}(t))_{t\ge0}$ to $\mathcal{X}(t) := \I_{\{ t \ge T_1 \}} + \I_{\{ t \ge T_1 + T_2 \}}$,  with joint density of $T_1, T_2$ denoted as $f_{T_1,T_2}$. Now the transition from $s_1:=0$ to $s_2:=1$ is not of interest, because the duration equals the deterministic time. Accordingly our interest is only in $s_3:=2$ and $\alpha(t, d):=\alpha_{12}(t, d) = \lim_{h\rightarrow 0+} \frac 1h P(\mathcal{X}(t+h) = 2 | \mathcal{X}(t-) = 1, T_1 = t - d)$ for $0< d< t$ (simplifying \eqref{alpha_def.leam}). Figure \ref{Beispielpfad} illustrates an outcome.

\begin{figure}[h]
	\centering
	\begin{tikzpicture}
		\draw[->] (-0.2,0) -- (7,0) node[below] {$t$};
		\draw[->] (0,-0.2) -- (0,2.5) node[left] {$\mathcal{X}(t)$};
		
		\draw (0.1, 0) -- (-0.1, 0) node[left] {$s_1=0$};
		\draw (0.1, 1) -- (-0.1, 1) node[left] {$s_2=1$};
		\draw (0.1, 2) -- (-0.1, 2) node[left] {$s_3=2$};
		
		\draw (2, 0.1) -- (2, -0.1) node[below] {$T_1$};
		\draw (5, 0.1) -- (5, -0.1) node[below] {$T_1 + T_2$};
		
		\draw[black, very thick] (0, 0) -- (2, 0);
		\draw[black, very thick] (2, 1) -- (5, 1);
		\draw[black, very thick] (5, 2) -- (7, 2);
		
		\node at (0, 0) [circle, fill = black, scale = 0.4] {};
		\node at (2, 1) [circle, fill = black, scale = 0.4] {};
		\node at (5, 2) [circle, fill = black, scale = 0.4] {};
	\end{tikzpicture}
	\caption{Realised history of progressive process $\mathcal{X}$.}
	\label{Beispielpfad}
\end{figure}
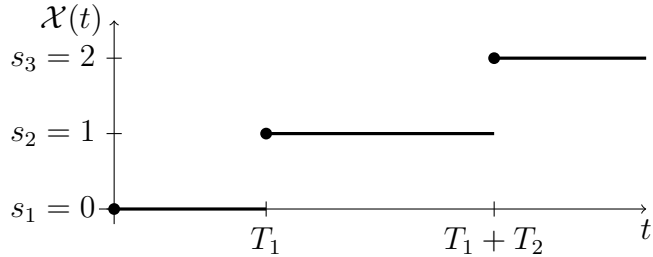

 A short calculation using \cite{Witting.1985}, Theorems 1.122 and 1.126, yields  
\begin{equation} \label{intensities_through_densities.lem}
\alpha(t, d)=\frac{f_{T_2|T_1 = t-d}(d)}{S_{T_2|T_1 = t-d}(d)}.
\end{equation}

Because of its relevance in practice (including our data example), we account for independent right-censoring at a random time $U$ and define $Z_1 := T_1 \wedge U$, as well as $Z_2 := (T_1 + T_2) \wedge U$. As usual \cite[see e.g.][]{keidinjens}, we only use the count of an uncensored transition $N(t):= \I_{\{Z_2 \le t, U \ge T_1 + T_2\}}$. The indicator of observable risk becomes $Y(t) :=\I_{\{X_{t-}=1,  U \ge t \}}= \I_{\{ Z_1 < t \le Z_2 \}} = \I_{\{ T_1 < t \le Z_2 \}} = \I_{\{ T_1 < t \le T_1 + T_2 \} \cap \{U \ge t\}}$. And the duration of $\mathcal{X}$ in state $1$ is modified so as to only contain information if $N$ is under risk of jumping, $D(t) := (t-T_1) Y(t)$.

As a component of counting process theory, an important part of a process' model is its filtration: 
\begin{equation} \label{deffilt}
\mathfrak F_t:=\sigma(N(u), D(u+), Y(u+),  u\le t)
\end{equation} 
The filtration needs to fulfil the (regular) Dellacherie conditions \cite[see][Definition 1.2.4]{Fleming.2011} in order to allow us the use of \cite{Linton.1995}, later on, in  Section \ref{esttest}.
Of the conditions, only the right-continuity is to show. Now define the almost surely piecewise constant $D_p(t):= T_1 Y(t)$ and note that $\mathfrak F_t  = \sigma(N(u), D_p(u+), Y(u+), u\le t)$ due to $D_p(t) = tY(t) - D(t)$, as a function of the linear $D$. The right-continuity follows from Theorem 4.2.3 of \cite{Fleming.2011}. Similar arguments for the semi-Markov process are omitted here.

\subsection{Intensity process}
In order to show asymptotic normality of an estimator for $\hat{\alpha}$ in the next Section \ref{sec31}, we need the Doob-Meyer decomposition of $N$ into compensator and martingale. For the martingale we will then need to show a Lindeberg-L\'{e}vy condition.
	
\begin{Sa} \label{martingale_property.th}
The intensity process of $N$ is $\alpha(t, D(t)) Y(t)$, note especially that it is predictable. Furthermore, $M(t) := N(t) - \int_0^t \alpha(s, D(s))Y(s)ds$ is a martingale. Both statements are with respect to $\mathfrak F_t$. 
\end{Sa}

\begin{proof}
We start with the progressive model and without censoring. Predictability easily follows by the left-continuity of $Y$ and $D$, as well as the continuity of $\alpha$. Without censoring we have the following simplifications of the progressive model $Y(t) =\I_{\{\mathcal{X}(t-)=1 \}}= \I_{\{ T_1 < t \le T_1 + T_2\}}$ and $N(t) = \I_{\{T_1 + T_2 \le t\}}$.	
Let be $0<s<t$. To begin with, note that the increments of $M$ are given by
\begin{align} \label{dM.eq}
	M(t) - M(s) = \I_{\{ s<T_1 + T_2 \le t \}} - \int_s^t Y(u) \alpha_{12}(u, D(u)) du.
\end{align}
We decompose $\Omega$  into the disjunct sets `The process jumps twice until time $s$.', `The process has no jump until time $s$.' and `The process jumps once until time $s$.' and show that 
the conditional expectation of $M(t) - M(s)$, given $\mathfrak{F}_s$, is zero separately:  
\begin{equation} \label{EdM_comp.eq}
	\begin{split}
		\IE[M(t) - M(s)|\mathfrak{F}_s] =  \I_{\{T_1 + T_2 \le s\}} \IE[M(t) - M(s)|\mathfrak{F}_s] \Big| _{\{T_1 + T_2 \le s\}}\\
		+  \I_{\{T_1>s\}} \IE[M(t) - M(s)|\mathfrak{F}_s] \Big|_{\{T_1>s\}}  
		+  \I_{\{T_1 \le s < T_1 + T_2\}}\IE[M(t) - M(s)|\mathfrak{F}_s] \Big|_{\{T_1 \le s < T_1 + T_2\}}
	\end{split}
\end{equation}
For a random variable $X:(\Omega, \mathfrak{A},P)\rightarrow (\Xi, \mathfrak{B})$ and a set $A\subset\Omega$ define $X|_A:A\rightarrow \Xi$ by $X|_A(\omega) = X(\omega),~\omega\in A$.
On the first set $\{T_1 + T_2 \le s\}$, it is obviously $\I_{\{ s<T_1 + T_2 \le t \}} = 0$. On this set  $Y(u)\equiv 0$ also holds. Hence, by inserting into \eqref{dM.eq}, we have $\I_{\{T_1 + T_2 \le s\}}\IE[M(t) - M(s)|\mathfrak{F}_s] \Big| _{\{T_1 + T_2 \le s\}}=0.$
On the second set $\{T_1>s\}$, the conditional expectation must be constant, in the argument $\omega$. We can see this, because, by the factorisation lemma \cite[see e.g.][Theorem 11.7]{bauer1992}, the conditional expectation can be represented as the composition of a measurable function with $(N(u), D(u+), Y(u+))_{u\le s}$. These random variables are all identical for all $\omega\in \{ T_1>s \}$. 
Thus, denote $\IE[\I_{\{ s<T_1 + T_2 \le t \}}|\mathfrak{F}_s] \Big|_{\{T_1>s\}}=:c$. For this $c$ holds, by definition of the conditional expectation, $c P(T_1 > s) = \int_{\{T_1 > s \}} cdP = \int_{\{ T_1 > s \}} \I_{\{ s< T_1 + T_2 \le t \}}dP\\ = P(s<T_1, T_1+T_2\le t)$. 
Furthermore, for $\IE\left[ \int_s^t Y(u) \alpha_{12}(u, D(u)) du|\mathfrak{F}_s \right] \Big|_{\{T_1>s\}}=:\tilde c$ we have:
\begin{align*}
	\tilde c &P(T_1 > s) = \int_{\{ T_1 > s \}} \tilde cdP = \int_{\{ T_1 > s \}} \int_s^t Y(u) \alpha_{12}(u, D(u)) dudP\\
	&\stackrel{(i)}{=}\IE_{T_1} \left[ \IE \left[ \int_s^t \I_{\{ s<T_1<u\le T_1 + T_2\}} \alpha_{12}(u, u-T_1) du \Big| T_1=t_1\right]  \right]\\ 
	&\stackrel{(ii)}{=} \IE_{T_1} \left[ \int_s^t \I_{(s,u)}(t_1) \alpha_{12}(u, u-t_1)  P \left( T_1 + T_2 \ge u\Big| T_1=t_1 \right)  du \right]\\
	&\stackrel{(iii)}{=}\IE_{T_1} \left[ \int_s^t \I_{(s,u)}(t_1) \frac{f_{T_2|T_1=t_1}(u-t_1)}{S_{T_2|T_1=t_1}(u-t_1)} S_{T_2|T_1=t_1}(u-t_1)  du \right] \\
	&\stackrel{(iv)}{=}\IE_{T_1} \left[\I_{(s,\infty)}(t_1) \int_{t_1}^t f_{T_2|T_1=t_1}(u - t_1)   du \right]\\
	&\stackrel{(v)}{=}\IE_{T_1} \left[ \I_{(s,\infty)}(t_1) P(T_1+T_2 \le t |T_1 = t_1)  \right]
	\stackrel{(vi)}{=}P(s<T_1, T_1+T_2\le t)
\end{align*}

For (i), as a notation $\IE_{T_1}[\dots]:=\int \dots dP_{T_1}(t_1)$. Furthermore, the integration set is formulated as an indicator function, and then joined with $Y(u)$. Accordingly, the outer integral can be written as an expectation, and due to the law of the iterated expectation, results in a conditional expectation. The argument $D(u)$ can be replaced by $u-T_1$, because both are equal on $\{Y(u)\ne0\}$. For (ii), we slip the conditional expectation into the integral, use the multiplication theorem for conditional expectations, in order to slip the $T_1$-measurable factor outside the expectation. The remaining expectation becomes a conditional probability. For (iii),  due to  \eqref{intensities_through_densities.lem}, we have - using Theorems 1.122 and 1.126 of \cite{Witting.1985} for the first identity: 
	\begin{equation} \label{cond_surv.eq}
	\begin{split}
		P(T_1 + T_2 \ge u| T_1 = t_1)&= \int \I_{[u-t_1, \infty)}(t_2) f_{T_2|T_1=t_1}(t_2) dt_2 \\&= \int_{u-t_1}^\infty f_{T_2|T_1=t_1}(t_2) dt_2 = S_{T_2|T_1=t_1}(u-t_1)
	\end{split}
\end{equation}
 Here, and from now on, integrals without borders extend over the interval $[0,1]$. For (iv), the conditional survival functions cancel,  and with $\I_{(s,u)}(t_1) = \I_{(s,\infty)}(t_1)\I_{(t_1,\infty)}(u)$ the second indicator function is formulated as an integral set. Now (v) is due to the following calculation, similar to (iii):
	\begin{align*}
	P&(T_1 + T_2 \le t| T_1 = t_1) = \int \I_{[0,t]}(t_1 + t_2) f_{T_2|T_1=t_1}(t_2)dt_2 \\&= \int_0^{t - t_1} f_{T_2|T_1=t_1}(t_2) dt_2 = \int_{t_1}^{t} f_{T_2|T_1=t_1}(u - t_1) du
\end{align*}
In the last step, we substitute $t_2 = u-t_1$. Finally, (vi) results from 
	\begin{align*}
		\IE_{T_1} &\left[ \I_{(s,\infty)}(t_1) P(T_1+T_2 \le t |T_1 = t_1)  \right]	=\IE \left[ \I_{\{T_1>s\}} P(T_1+T_2 \le t |T_1)  \right]\\
		&= \IE \left[ \IE \left[ \I_{\{T_1>s, T_1+T_2 \le t \} }|T_1  \right] \right]=P(s<T_1, T_1+T_2\le t),
	\end{align*}
so that $c = \tilde c$ and hence $\I_{\{T_1>s\}}\IE[M(t) - M(s)|\mathfrak{F}_s] \Big|_{\{T_1>s\}}=0.$
It remains to investigate the conditional expectation on the third set, $\{T_1 \le s < T_1 + T_2\}$ and we first simplify $\IE[\cdot| \mathfrak{F}_s]\big|_{\{T_1\le s<T_1 + T_2\}} = \IE[\cdot| T_1, \I_{\{T_1+T_2 > s\}}]\big|_{\{T_1\le s<T_1 + T_2\}}.$ The heuristic reason for that is the following: When we know that event  $\{T_1\le s<T_1 + T_2\}$ has occurred, the only information needed, to reconstruct the process until time $s$, is the outcome of $T_1$. The indicator function must be added to the condition, in order to render the set $\{T_1\le s< T_1+T_2\}$ measurable, relative to the $\sigma$-Algebra of the condition. The idea is formalised in Lemma \ref{Bedingungstausch.lem}. To verify its conditions, recall the definition of $\mathfrak{F}_s$ in \eqref{deffilt}. Now, it is $\{T_1 \le s < T_1 + T_2\} = \{Y(s) = 1\} \in\mathfrak{F}_s$ as well as $\{T_1 \le s < T_1 + T_2\} = \{T_1\le s, \I_{\{T_1+T_2 > s\}} = 1\} \in\sigma(T_1, \I_{\{T_1+T_2 > s\}})$. For the application of  Lemma \ref{Bedingungstausch.lem} we still need to check that $\sigma(T_1, \I_{\{T_1+T_2 > s\}}|_{\{T_1 \le s < T_1 + T_2\}}) = \sigma((N(u), D(u+), Y(u+), u\le s)|_{\{T_1 \le s < T_1 + T_2\}})$, using Lemma \ref{y_messbar.lem}. 
We define the $Y(u)$-measurable random variable $W_u:=u \I_{\{Y(u) = 1\}} + \infty \I_{\{Y(u) = 0\}}$. For all $\omega \in \{T_1 \le s < T_1 + T_2\}$ the following holds on the one hand
$T_1(\omega) = \inf\{u \le s| Y(u, \omega) = 1\} = \inf\{W(u,\omega)|u\le s\} = \inf\{W(q,\omega)|q\le s, q\in\IQ\}$ (as well as $\I_{\{T_1+T_2 > s\}}(\omega)\equiv 1$),  
and on the other hand $N(u, \omega) = 0$, $Y(u, \omega) = \I_{\{T_1(\omega) \le u\}}$ and $D(u, \omega) = (u - T_1(\omega)) Y(u, \omega)$ (all for $u\le s$). Because especially $\inf:\bar\IR^{\IN}\rightarrow\bar\IR$ is Borel-measurable, it follows that Lemma \ref{y_messbar.lem} may be applied. We inspect firstly, 
\begin{align*}
	\IE \big[ & \I_{\{s < T_1 + T_2 \le t\}} |T_1 = t_1 \big] = \int \I_{(s,t]}(t_1 + u) f_{T_2|T_1 = t_1}(u)du \\
	&= \int \I_{(s - t_1,t - t_1]}(u) f_{T_2|T_1 = t_1}(u) du = \int_{s-t_1}^{t-t_1}  f_{T_2|T_1 = t_1}(u) du \\
	&= \int_{(s-t_1)\vee 0}^{t-t_1}  f_{T_2|T_1 = t_1}(u)du = \int_{s\vee t_1}^{t}  f_{T_2|T_1 = t_1}(u - t_1)du,
\end{align*}
where again, Theorems 1.122 and 1.126 of \cite{Witting.1985} are used, together with the a.s.-positivity of $T_2$. Additionally, it holds that 
\begin{align*}
	\IE \big[ \I_{\{s < T_1 + T_2 \le t\}} |T_1 = t_1, \I_{\{T_1+T_2 > s\} }= 1 \big] = \frac{\IE \big[ \I_{\{s < T_1 + T_2 \le t\}} |T_1 = t_1\big]}{P(T_1+T_2 > s|T_1 = t_1)},
\end{align*}
because $\{s < T_1 + T_2 \le t\} \subset \{T_1+T_2 > s\}$. The more general idea here is that for $A\subset B$
\begin{align*}
	\IE[\I_A | T=t, \I_B = 1] = P( A | T=t, B) = \frac{P(A\cap B| T=t)}{P(B|T=t)} =  \frac{\IE[\I_A| T=t]}{P(B|T=t)}.
\end{align*}
Conditional on $\I_{\{T_1+T_2 > s\} }= 0$, the result is zero with an analogous argumentation. Furthermore we have: 
\begin{align*}
	\IE &\left[ \left. \int_s^t Y(u) \alpha(u, D(u)) du \right| T_1 = t_1 \right] \stackrel{(i)}{=}  \int_s^t \alpha(u, u - t_1) \IE\big[ \I_{\{T_1 < u \le T_1 + T_2\}} | T_1 = t_1\big] du \\
	&\stackrel{(ii)}{=} \int_s^t \I_{(t_1,\infty)}(u) \frac{f_{T_2|T_1=t_1}(u - t_1)}{S_{T_2|T_1=t_1}(u - t_1)} P(u \le T_1 + T_2 | T_1 = t_1) du\\
	&\stackrel{(iii)}{=} \int_{s \vee t_1}^t \frac{f_{T_2|T_1=t_1}(u - t_1)}{S_{T_2|T_1=t_1}(u - t_1)}S_{T_2|T_1=t_1}(u - t_1) du = \int_{s \vee t_1}^t f_{T_2|T_1=t_1}(u - t_1) du
\end{align*}
For (i) we interchange the integrations and slip the $T_1$-measurable factor $\alpha$ out of the conditional expectation. For (ii), we again slip a $T_1$-measurable factor and use \eqref{intensities_through_densities.lem}. For (iii) see \eqref{cond_surv.eq}.
Again, we have 
\begin{multline*}
	\IE \left[ \left.  \int_s^t Y(u) \alpha_{12}(u, D(u)) du \right| T_1 = t_1, \I_{\{T_1+T_2 > s\} }= 1 \right] \\ = \frac{\IE \left[ \left.  \int_s^t Y(u) \alpha_{12}(u, D(u)) du \right| T_1 = t_1 \right] }{P(T_1+T_2 > s|T_1 = t_1)},
\end{multline*}
best seen in the above formula after step (i), because $\{T_1 < u \le T_1 + T_2\} \subset \{T_1+T_2 > s\}$ for all $u>s$. As a consequence, we have all in all $$   \IE \big[ M(t) - M(s) | T_1 = \cdot, \I_{\{T_1+T_2 > s\} }=\cdot \big]\Big|_{\{T_1 \le s < T_1 + T_2\}}  = 0 $$ and thereby also $\I_{\{T_1 \le s < T_1 + T_2\}}\IE[M(t) - M(s)|\mathfrak{F}_s] |_{\{T_1 \le s < T_1 + T_2\}} = 0.$ 
This ends the proof for the progressive model with completely observed histories. The right-censored design requires some additional, but elementary, calculations with the details provided in Appendix \ref{withcens}. For the semi-Markov process, the main effort is that transitions can occur arbitrarily often and the details are in Appendix \ref{semimarkmart}.  
\end{proof}

\subsection{Assumptions and hypotheses} \label{sec2_3}

We will need $f(t,d)$, the density of $D(t)$ conditional on $\{Y(t)=1\}$. Further, we denote $\mathrm{y}(t) := \IE Y(t) = P(Y_i(t) = 1)$ and write short $x=:(t,d)$ and $\varphi(x) = \varphi(t,d) :=  f(t,d) \mathrm{y}(t)$. Let be $\mathscr X = \mathscr X _1 \times \mathscr X _2$ a two-dimensional subset of $[0,1]^2$, so that each $\mathscr X _i$ is a compact interval. For  $x\in(0,1)^2$, let $\mathcal N := [x-\epsilon, x + \epsilon] \subset (0,1)^2$ define  the neighbourhood (with $\epsilon \in (0,1)^2$). It is helpful to label some assumptions. 
\begin{itemize}
	\item[(D)] The function $\varphi$ is strictly positive on $\mathcal N$.
	\item[(D')] It holds $\inf_{x\in\mathscr X} \varphi(x)> 0$.
	\item[(S)] The function $\alpha$ is twice and the function $\varphi$ once differentiable on $\mathcal N$.
	\item[(S')] Both $\alpha$ and $\varphi$ are continuous on $\mathscr X$;
\end{itemize}

For a test on the omission of the second `covariate' in $\alpha(t,D(t))$ we define  $\mathcal B := \big\{\beta:[0,1]^2 \rightarrow[0,\infty) \big| \beta(t ,d_1)=\beta(t, d_2), \forall t,d_1,d_2 \in[0,1]\big\}$. 
Extending the definition $x=(t,d)$ to $X(t):=(t,D(t))$, the hypotheses are:
\begin{equation} \label{hy}
	\begin{split}
		H_0:   \exists \; \text{function} \; \beta \in \mathcal B, \; \text{such that} \; P\big(\alpha(X(t)) = \beta(X(t)), t\in[0,1]\big) = 1  \\
				H_1:  \forall \; \text{functions} \; \beta \in \mathcal B \; \text{holds} \; P\big(\alpha(X(t)) = \beta(X(t)), t\in[0,1]\big) < 1
	\end{split}
\end{equation}

%\newpage

\section{Test for Duration Dependence} \label{esttest}

We now assume a simple sample $(N_1,D_1,Y_1), \ldots, (N_n,D_n,Y_n)$. Note that from here on, arguments hold for both the progressive model with censoring and the semi-Markov model, because by Theorem \ref{martingale_property.th} $\alpha(t, D(t)) Y(t)$ is the intensity process for both definitions of $N$, $Y$ and $D$.  

\subsection{Smoothing conditions and estimators} \label{sec31}
We will estimate $\alpha$ by kernel smoothing. To this end, let $k$ be a one-dimensional density with moments $\kappa_1:=\int_{-1}^1v^2k(v)dv$ and $\kappa_2 := \int_{-1}^1 k(v)^2 dv$. Now, for $b>0$, let be $k_b(\cdot) := b^{-1}k(\cdot / b)$, bivariate  $K(u) := k(u_1) k(u_2)$ and $K_b(u) := k_b(u_1) k_b(u_2).$ \\
Again, it is helpful to label some assumptions.
\begin{itemize}
	\item[(K)] The kernel function $k$ is supported on $[-1, 1]$, is symmetric around zero and continuous. 
		\item[(K')] It holds (K). Additionally $k$ is Lipschitz-continuous, i.e. there exists a $C>0$, so that  $|k(u)-k(v)|\le C |u-v|$ for all $u,v$.
		\item[(K'')] It holds (K'). Furthermore let be $k^2$ Lipschitz-continuous with constant $\tilde C>0$.
	\item[(B)] For $n\rightarrow \infty$ holds, $nb^2\rightarrow \infty$ and $b\rightarrow 0$.
	\item[(B')] For $n\rightarrow \infty$ holds, $nb^4\rightarrow \infty$ and $b\rightarrow 0$.  
		\item[(\~ B)] For $n\rightarrow \infty$ holds, $nb^6\rightarrow 0$.
\end{itemize}

Under the Markovian assumption, i.e. not taking the duration in state $1$ until time $t$ into account, \cite{Hjort.1994} recommends estimating the intensity $\alpha$ by 
\begin{equation} \label{hjort1994}
	\widehat \beta(t):= \frac{\sum_{i=1}^n \int k_b(t - s) dN_i(s)}{\sum_{i=1}^n \int k_b(t - s)Y_i(s) ds}.
\end{equation}
In the semi-Markovian model, we use the estimator of \cite{Linton.1995}:
\begin{equation} \label{unrestest}
	\widehat \alpha(t,d):= \frac{\sum_{i=1}^n \int k_b(d - D_i(s)) k_b(t - s) dN_i(s)}{\sum_{i=1}^n \int k_b(d - D_i(s)) k_b(t - s)Y_i(s) ds}
\end{equation}

In the sense of \cite{Jones.1994}, it is an `external' estimator, meaning that the normalisation is realised outside of the sum (in the numerator). Hence, it is similar in structure to the Nadaraya-Watson estimator for regression functions. Kernel smoothing of the Nelson-Aalen estimator \cite[see e.g.][]{Ramlau-Hansen.1983} and generalisations thereof with covariates \cite[see e.g.][]{McKeague.1990} are, by contrast, `internal' estimators.

\subsection{Multivariate normality of estimator}

We now study the asymptotic normality of the (unrestricted) estimator \eqref{unrestest} under the alternative hypothesis. The processes $M_i, i=1,...,n$, as of Theorem \ref{martingale_property.th}, are square-integrable local on the interval $[0,1]$. 
With the definition 
\begin{align} \label{alpha_star.eq}
	\alpha^*(x) =  \frac{\sum_{i=1}^n \int K_b(x-X_i(s))\alpha(X_i(s))Y_i(s)ds}{\sum_{i=1}^n \int K_b(x-X_i(s))Y_i(s)ds}
\end{align}
we may decompose the difference $(\hat\alpha - \alpha)(x)$ into two summands,
\begin{align} \label{alpha_decomp.eq}
	(\hat\alpha - \alpha)(x) = (\hat\alpha - \alpha^*)(x) + (\alpha^* - \alpha)(x) = \frac{\mathscr{V}_x + \mathscr B_x}{\mathscr{C}_x},
\end{align}
where
$\mathscr{V}_x := \frac 1n \sum_{i=1}^n \int K_b(x-X_i(s))dM_i(s)$, $\mathscr{C}_x := \frac 1n \sum_{i=1}^n \int K_b(x-X_i(s))Y_i(s)ds$ and $\mathscr{B}_x := \frac 1n \sum_{i=1}^n \int K_b(x-X_i(s))\left[ \alpha(X_i(s)) - \alpha(x) \right]Y_i(s)ds$. Such decompositions are typical in the asymptotic analysis of kernel smoothing, primarily to decouple the bias term. Here, \cite{Linton.1995} denote the first as the `variable' term and the second as `stable'.

As the usual conditions for $\mathfrak{F}_t$ are fulfilled and due to Theorem \ref{martingale_property.th}, the following holds due to Theorem 1 in \cite{Linton.1995}.
\begin{Sa} \label{satz3_1}
	Under assumptions (D), (S), (K) and (B) holds:
	\begin{itemize}
		\item[(a)] $n^{1/2} b \left( \widehat \alpha(t,d) -\alpha^* (t,d) \right) \Rightarrow \mathcal{N} \left[ 0, \sigma_{t, d}^2 \right],$ where $\sigma_{t,d}^2 := \kappa_2^{2} \frac{\alpha(t,d)}{\varphi(t,d)}$
		\item[(b)] $b^{-2} \left( \alpha^*(t,d) - \alpha(t,d) \right) \overset P \longrightarrow c(t, d),$ where
		\begin{align*}
			c(t,d) := \kappa_1 \Bigg[ \frac{(\partial \alpha(t,d)/\partial t)(\partial \varphi(t,d) / \partial t)}{\varphi(t,d)} + \frac{\partial^2 \alpha / \partial t^2}{2} \\+ \frac{(\partial \alpha(t,d)/\partial d)(\partial \varphi(t,d) / \partial d)}{\varphi(t,d)} + \frac{\partial^2 \alpha / \partial d^2}{2}\Bigg]
		\end{align*} 
		\item[(c)] $\widehat \sigma_{t,d}^2 \overset P \longrightarrow \sigma_{t,d}^2,$ where 
		$\widehat \sigma_{t,d}^2 := \frac{n^{-1}b^{2} \sum_{i=1}^n \int k_b^2(d-D_i(s))k_b^2(t-s)dN_i(s)}{\left( n^{-1} \sum_{i=1}^n \int k_b(d - D_i(s)) k_b(t - s)Y_i(s) ds \right)^{2}}$.
	\end{itemize}
\end{Sa}

% Aber auch eine  multivariate Formulierung (nach Satz \ref{fdd_conv.th}), eine gleichmäßige Aussage (nach Satz %\ref{global.th}) oder eine um $\alpha_{12}$ zentrierte Konvergenzsaussage (nach Korollar %\ref{alpha_centered_conv.corollary}) sind möglich. 

This result, for one point $x$, should now be generalized to arbitrarily (but finitely) many. We find that - asymptotically - estimators at different points are independent. Similar results hold for kernel density estimation and kernel regression \cite[see e.g.][p. 88+120]{Nadaraya.1989}. However, first note that a short calculation yields, for a real-valued function $g$:
\begin{equation} \label{EV.lem} 
	\IE \int g(Z_i(s), s) Y_i(s) ds = \int_{[0,1]^2} g(w) \varphi(w) dw
\end{equation}

Also, we need as consequence of Proposition 1 in \cite{Linton.1995}, using the Cramer-Wold device \cite[see e.g.][Theorem 29.4]{Billingsley.2012}: 
\begin{Lem} \label{mvmlt.lem}
	For all $n\in\IN$, let $H_{i,j}^{(n)}, i=1,...,n, j=1,...,\zeta,$ be predictable stochastic processes on the interval $[0,1]$, where $\zeta\in\IN$ is independent of $n$. For $n\rightarrow \infty$ let the following assumptions be fulfilled.  
	\begin{itemize}
		\item[(G1')] For a positive definite, symmetric $\zeta\times \zeta$-matrix $A = \{a_{jk}\}_{j,k=1,...,\zeta}$ holds $$\sum_{i=1}^n \int H_{i,j}^{(n)}(s)H_{i,k}^{(n)}(s) d\langle M_i \rangle (s) \longrightarrow_P a_{jk}, ~ j,k=1,...,\zeta.$$
		\item[(G2')] For all $\epsilon>0$ let  $\sum_{i=1}^n \int \left\{ H_{i,j}^{(n)} \right\} ^2(s) \mathbbm{1}_{\left\{ \left| H_{i,j}^{(n)}(s) \right|>\epsilon \right\}} d\langle M_i \rangle (s) \longrightarrow_P 0, ~ j=1,...,\zeta$ hold.
	\end{itemize}
	Then it is $ \left[ \sum_{i=1}^n \int H_{i,j}^{(n)}(s)d M_i (s) \right]_{j=1,...,\zeta} \Longrightarrow  N(\mathbf{0}, A).$
\end{Lem}

This is a $\zeta$-dimensional generalisation of Proposition 1 in \cite{Linton.1995}. The relation between them is similar to that between the classic process-valued limit theorem of \citet{Rebolledo.1980} and Theorem I.2 in \cite{Andersen.1982}. The proof of the following is in Appendix \ref{furthres}.

\begin{Sa} \label{fdd_conv.th}
	For any pairwise different points $x_1,...,x_{\zeta} \in (0,1)^2$ let Assumptions (D) and (S) be fulfilled. Additionally let hold (K) and (B). Then it is   
	$$\sqrt{n}b(\hat\alpha - \alpha^*)(\bd{x})\Longrightarrow N \left( \bd{0}, diag\left( \kappa_2^2 \frac{\alpha(x_j)}{\varphi(x_j)}, j=1,...,\zeta \right) \right).$$
\end{Sa} 

As for confidence intervals and testing, the asymptotic normality around $\alpha$, instead of $\alpha^*$, is needed, under-smoothing, namely (\~B), can be applied in order to let the stable term converge faster to zero than the width of the confidence interval. 
\begin{Kor} \label{alpha_centered_conv.corollary}
It is $\sqrt{n}b(\hat\alpha - \alpha)(\bd{x})\Longrightarrow N \left( \bd{0}, diag\left( \kappa_2^2 \alpha(x_j)/\varphi(x_j), j=1,...,\zeta \right) \right)$, when (\~B) is added to the assumptions of Theorem \ref{fdd_conv.th}. 
\end{Kor}
\begin{proof}
	We note that
$\sqrt{n}b(\hat\alpha - \alpha)(\bd{x}) = \sqrt{n}b(\hat\alpha - \alpha^*)(\bd{x}) + \sqrt{n}b b^2 b^{-2} (\alpha^* - \alpha)(\bd x)$.
The first summand converges, due to Theorem \ref{fdd_conv.th}, towards the required distribution. According to \cite{Linton.1995}, Theorem 1, $b^{-2}(\alpha^* - \alpha)(\bd x)$ converges towards a deterministic vector, because the stochastic convergence of a vector is equivalent to the stochastic convergence of its coordinates. The leading factor $\sqrt{n}b b^2 = (nb^6)^{1/2}$ converges toward zero, due to (\~B).
\end{proof}

\subsection{Grid search test}

We construct here a consistent and powerful test for the hypotheses \eqref{hy}, based on \cite{Gozalo.1993}. With an unrestricted estimator of $\alpha$ as in \eqref{unrestest}, with an estimator of $\alpha$ under hypothesis $\hat\beta$ from \eqref{hjort1994} and with estimator of the standard error of $\hat{\alpha}(x)$, namely $\hat\sigma_x$ (from Theorem \ref{satz3_1}), we now define for the test statistic
\begin{align*}
	S_n(x) := \left( nb^2 \right)^{\frac 12} \frac{\hat\alpha(x) - \hat\beta(x)}{\hat\sigma_x}.
\end{align*}

We will need to ensure that the weak point-wise consistency of $\hat\sigma_x^2$ (given in Theorem \ref{satz3_1} (c)) also holds uniformly, and make use of Assumption (K''). With the proof in Appendix \ref{unifconvse} and recalling the definition of $\mathscr X$ from Section \ref{sec2_3}, we state: 

\begin{Sa} \label{uniform_var.th}
	Under the Assumptions (D'), (S'), (K'') and (B') holds $\sup_{x\in\mathscr X}|\hat\sigma_x^2 - \sigma_x^2| \longrightarrow_P0.$
\end{Sa}

 For the test to have asymptotic level $\gamma$, we need, under hypothesis, $\hat\beta$  to be consistent with a faster convergence rate than that, under alternative hypothesis, of $\sqrt{n}b$ for $\hat\alpha$ (see Corollary \ref{alpha_centered_conv.corollary}). In fact, Theorem 1 from \cite{Linton.1995} already suggests this (for the dimension of the covariate degenerated to zero), and it is given explicitly in \cite{Hjort.1994}. Hence, we have a real sequence of numbers $a_n^\beta$, with $(\sqrt{n}b)/ a_n^\beta\rightarrow 0$ under $H_0$ for $x\in\{x^1,...,x^{\zeta}\}$, with $a_n^\beta \left( \hat\beta(x) - \alpha(x) \right) \overset{d}{\rightarrow} \xi$, 
	where $\xi$ is a random variable with expectation $0$ and variance $\sigma_\beta^2(x)$, which can be consistently estimated by $\hat\sigma_\beta^2(x)$.

The test is now to reject $H_0$ if, on a grid of points, one searches for the largest difference of $\hat\alpha(x)$ and $\hat\beta(x)$, and rejects if the distance exceeds the critical value. By doing this, the test statistic's distribution does not depend on the number of grid points. We follow \cite{Gozalo.1993}, who finds that using a random grid is powerful. Let $f_{\tilde X}$ denote a density on $\mathscr X$. Further let $\{\tilde X^j\}_{j\in\IN}$ be a sequence of independent  random vectors, all distributed with $f_{\tilde X}$, and independent of the data. Furthermore, the asymptotic analysis is simplified by implementing a penalisation which ensures that earlier ages are more influential and that the maximum is asymptotically unique. 
Also let $\{\zeta_n\}_{n\in\IN}$ be a sequence of natural numbers with $\zeta_n \rightarrow \infty$, and $\zeta_n = o((nb^2)^\delta)$, for some $\delta\in(0,1)$, for $n\rightarrow\infty$. Now we define, for some $\eta>0$, 
$J := \arg\max_{1\le j\le \zeta_n}\left\{ S_n(\tilde X^j)^2 - \eta (nb^2)^\delta \mathbbm 1 _{\{j>1\}} \right\}$
 and $\hat X := \tilde X ^J.$ And as the penultimate prerequisite, by  \cite{Hjort.1994}, there is a function $\beta_0\in\mathcal B$, for which uniform consistency holds: 
\begin{equation} \label{astrich}
	\sup_{x\in\mathscr X} |\hat\beta (x) - \beta_0(x)| \rightarrow_P 0 \quad \text{for} \quad n\rightarrow\infty
\end{equation}

For the proof of the next Theorem \ref{Gozalo_2.th}, we state as an important part of it (with proof in Appendix \ref{furthres}).
\begin{Lem} \label{gozalo_problem.lem}
	Under $H_0$ and Assumptions (D'), (K''), (B'), (\~B) as well as (S) for all $x\in\mathscr X$ holds	$\frac 1 {\zeta_n} \sum_{j=1}^{\zeta_n} S_n(\tilde X^j)^2 = O_P(1)$.
\end{Lem}

\begin{Sa} \label{Gozalo_2.th}
	Let (D'), (K''), (B') and (\~B) hold, as well as (S) for all $x\in\mathscr X$. \\
	(a) Under $H_0$ holds, for $n\rightarrow\infty$, $P(J=1) \rightarrow 1$ and $S_n(\hat X)^2 \Longrightarrow \chi_1^2$. \\
	(b) Given $H_1$ and  
	\begin{align} \label{gozalo_2_ann_b.eq}
		P(\alpha(\tilde X)\ne\beta_0(\tilde X)) > 0,
	\end{align}
then, for $\{k_n\}_n=O((nb^2)^\delta)$, it is $P(S_n(\hat X)^2>k_n) \stackrel{n\rightarrow\infty}{\longrightarrow}  1$.
\\ 
 (c) Let $\{k_n\}_n$ be as in (b) and furthermore $P(\hat\beta\in\mathcal B) = 1$. Under $H_1$ and given that there is an $\epsilon> 0$ with $\inf_{\beta\in\mathcal B} P(|\alpha(\tilde X) - \beta(\tilde X) | > \epsilon) =: p > 0$, then again $P(S_n(\hat X)^2>k_n) \stackrel{n\rightarrow\infty}{\longrightarrow}  1$.
\end{Sa}

\begin{proof}
(a) For all $\mu>0$ holds 
\begin{align*}
	P\left( \max_{j=1,...,\zeta_n}S_n(\tilde X^j)^2  \ge \mu (nb^{d+1})^\delta \right) \le 	P\left(\frac1{\zeta_n} \sum_{j=1}^{\zeta_n}S_n(\tilde X^j)^2 \ge \mu\frac{(nb^{d+1})^\delta}{\zeta_n}   \right) \underset{n\rightarrow\infty}{\longrightarrow} 0.
\end{align*}
Firstly, one bounds the maximum by the sum and divides by $\zeta_n$. The convergence now follows with Lemma \ref{gozalo_problem.lem}, because the ratio $(nb^{d+1})^\delta/\zeta_n$ diverges to infinity. This is because in general, it holds  for $X_n=O_P(1)$, $a_n\rightarrow\infty$, that for each $\epsilon>0,  M>0$, there are $N_1>0$ and $N_2>0$ such that: For all $n\ge N_1$ holds $P(|X_n|>M) < \epsilon$; for all $n\ge N_2$ holds $a_n> M$. Finally, by this for all  $n\ge N_1 \vee N_2$ it holds that $P(|X_n|>a_n) < \epsilon$. Now, the first part of (a) results from 
\begin{align*}
	P(J\ne1) &= P(\exists j\in\{2,...,\zeta_n\}: S_n(\tilde X^j)^2 - S_n(\tilde X^1)^2 > \eta (nb^{d+1})^\delta) \\ &\le P(\exists j\in\{1,...,\zeta_n\}: S_n(\tilde X^j)^2 > \eta (nb^{d+1})^\delta) \\ &\le P\left( \max_{j=1,...,\zeta_n}S_n(\tilde X^j)^2  \ge \eta (nb^{d+1})^\delta \right)\underset{n\rightarrow\infty}{\longrightarrow} 0.
\end{align*}
The second part results, because $S_n(\tilde X^1)^2\Rightarrow \chi_1^2$ \cite[see][Theorem 2.3]{Gozalo.1993}). For (b), we note first that $\max_j |S_n(\tilde X^j)|$ 
\begin{equation} \label{S_n_ge.eq}
	\begin{split}
	\stackrel{(i)}{=} \max_j \frac{(nb^{d+1})^{1/2}}{\hat\sigma_{\tilde X^j}} \left| \hat\alpha(\tilde X^j) - \alpha(\tilde X^j) +\alpha(\tilde X^j) - \beta_0 (\tilde X^j) +\beta_0(\tilde X^j) - \hat\beta(\tilde X^j) \right| \\
	\stackrel{(ii)}{\ge}  \frac{(nb^{d+1})^{1/2}}{\max_j \hat\sigma_{\tilde X^j}} \left(\max_j \left|\alpha(\tilde X^j) - \beta_0 (\tilde X^j) \right| - \max_j \left|\hat\alpha(\tilde X^j) - \alpha(\tilde X^j) \right| \right. \\
	\left.  - \max_j \left| \beta_0(\tilde X^j) - \hat\beta(\tilde X^j) \right| \right)
\end{split}
\end{equation}
For (i), we include two valuable zeros in the definition of $S_n$. For (ii), we use the fact that for the norm $\| \cdot \|$ the inequality $\|x+y\|\ge \|x\| - \|y\|$ holds. In the sequel,  
\begin{align*}
	P(S_n(\hat X)^2 > k_n) &\stackrel{(i)}{\ge} P\left( \max_{j=1,...,\zeta_n} |S_n(\tilde X^j)| > \left( \eta ( n b^{d+1})^\delta + k_n \right)^{1/2} \right) \\
	&\stackrel{(ii)}{\ge} P\Bigg( \max_j |\alpha(\tilde X^j) - \beta_0(\tilde X^j)| > \max_j \hat \sigma _{\tilde X^j} \left( \frac{ \eta (n b^{d+1})^\delta + k_n  }{nb^{d+1}} \right)^{1/2} \\&~~~~~~~~+ \max_j \left|\hat\alpha(\tilde X^j) - \alpha(\tilde X^j) \right|+ \max_j \left| \beta_0(\tilde X^j) - \hat\beta(\tilde X^j) \right|\Bigg) \stackrel{(iii)}{\longrightarrow} 1.
\end{align*}
For (i), due to the definition of $\hat X$, $S_n(\hat X)^2 \ge \max_j S_n(\tilde X^j)^2 - \eta (n b^{d+1})^\delta$ holds. We insert this, move one term on the other side of the inequality and apply the square root. For (ii), we insert \eqref{S_n_ge.eq}. For (iii), all terms on the right hand side of the inequality in the probability converge in probability towards zero: The first because $\hat\sigma_x$ converges uniformly in $x$ by Theorem \ref{uniform_var.th} and $k_n=O((nb^{d+1})^\delta)$ (note Assumption (B'), which implies (B)); the second according to Theorem 2 in \cite{Linton.1995} (note that Assumption (S) is stronger than (S')), and the third according to \eqref{astrich}. The left side is now, due to \eqref{gozalo_2_ann_b.eq}, asymptotically bounded away from zero. For (c), due to an argument analogous to one in the proof of (b), it holds that $P(S_n(\hat X)^2 > k_n)$ is larger or equal to
\begin{equation} \label{S_n_ge_2.eq}
 P\Bigg( \max_j |\alpha(\tilde X^j) - \hat\beta(\tilde X^j)| > \max_j \hat \sigma _{\tilde X^j} \left( \frac{ \eta (n b^{d+1})^\delta + k_n  }{nb^{d+1}} \right)^{1/2} + \max_j \left|\hat\alpha(\tilde X^j) - \alpha(\tilde X^j) \right|\Bigg).
\end{equation}
For equally analogous reasons, it follows that the right hand side of the inequality, within the probability, converges to zero stochastically. Furthermore, we have for $\epsilon$ from the assumptions of the Theorem that  
\begin{align*}
	P\left( \max_j |\alpha(\tilde X^j) - \hat \beta (\tilde X^j)| > \epsilon \right) &= \int P\left( \left. \max_j |\alpha(\tilde X^j) - \hat \beta (\tilde X^j)| > \epsilon \right| \hat\beta = \beta \right) dP_{\hat\beta}(\beta) \\
	&=\int P\left( \max_j |\alpha(\tilde X^j) - \beta (\tilde X^j)| > \epsilon \right)dP_{\hat\beta}(\beta),
\end{align*}
where the second identity used the independence of $\tilde X^j$ from the data. Additionally holds the following
\begin{multline*}
	P\left( \max_j |\alpha(\tilde X^j) - \beta (\tilde X^j)| > \epsilon \right) = 1- P\left( |\alpha(\tilde X) - \beta (\tilde X)| > \epsilon \right)^{\zeta_n} \\ \ge 1- \sup_{\beta\in\mathcal B} P\left( |\alpha(\tilde X) - \beta (\tilde X)| > \epsilon \right)^{\zeta_n}  = 1 - (1-p)^{\zeta_n} \rightarrow 1,
\end{multline*}
with $p>0$ by the assumptions of the Theorem. Altogether, the left-hand side, within the probability in \eqref{S_n_ge_2.eq}, is asymptotically bounded away from zero and hence, \eqref{S_n_ge_2.eq} converges towards one.
\end{proof}

%\newpage

\section{Data example: Mortality with dementia} \label{app}

Germany's largest health insurance company `Allgemeine Ortskrankenkasse (AOK)' supplied us with a simple sample of a quarter million people, born before 1954 and observed between 01/01/2004 and 31/12/2013 \cite[see also][]{Weissbach.2021}. We count age $t$ from year 50 onwards in years; hence the youngest person is 50 years old at the beginning of the period. The oldest is 113, because we restrict the sample to people born from 1900 onwards. We concentrate here on the chronic disease dementia in the semi-Markovian progressive healthy-ill-dead model  \cite[see e.g.][Chapt. 3.3]{Andersen.2002}. 
The data contain date of birth, date of dementia onset and/or date of death, both or either if having occurred during the observation period.  Roughly 14\% of the observed persons had a dementia onset, i.e. 32,000 insured. Around 15\% of the people get lost to follow-up during the period and are hence right-censored.  Also right-censoring are those 52\% who are alive in 2015. Fifty-six percent of the data are women, because they live longer than men and are hence also relatively more affected by dementia. Nonetheless their onset of dementia is later than for men. We restrict our analysis to the sample of $n=130,168$ observed women.

We are interested in the transition intensity $\alpha(t, d)$ from `dementia' ($r$) to `dead' ($s$),  with $t$ as age and $d$ as time-since-dementia-onset, both in years (see \eqref{alpha_def.leam}). 
In terms of Section \ref{sect21}, $Z_1$ is the minimum of the age at dementia onset or the age at censoring. And $Z_2$ is the minimum age at death or censoring. 
Of the $130,168$ observations, only those $20,721$ women with dementia onset in the observation period enter the unrestricted estimator \eqref{unrestest} (see Figure \ref{surface.plot}, left), all others have $dN_i \equiv0$ and $Y_i\equiv0$. We use a triangular kernel $k(u) := \I_{[-1, 1]}(u) (1 - |u|)$. 
The bandwidth selection, including aspects of censoring is not considered here \cite[see e.g.][]{a-general-:2006,a-rule-of-:2006}. We use a fixed bandwidth, for the $t$-axis it is $b=2$ years. Because $d$ is on a shorter domain, we deviate slightly from the universal bandwidth in Section \ref{sec31}, and use in $d$-direction $b^d=1.33$ years.

\begin{figure}[htb] 
	\centering
	\includegraphics[height=0.25\textheight]{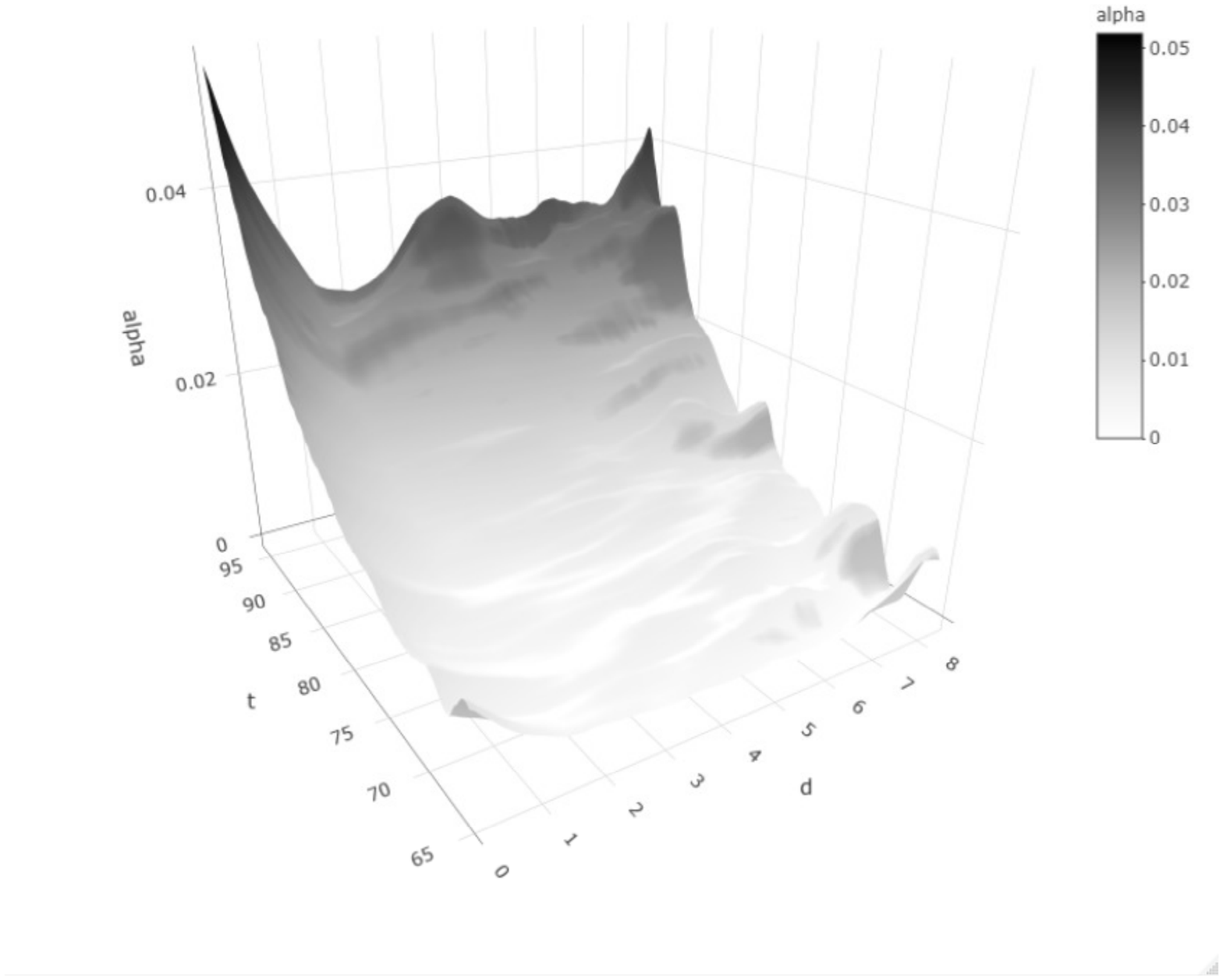}
	\includegraphics[height=0.2\textheight]{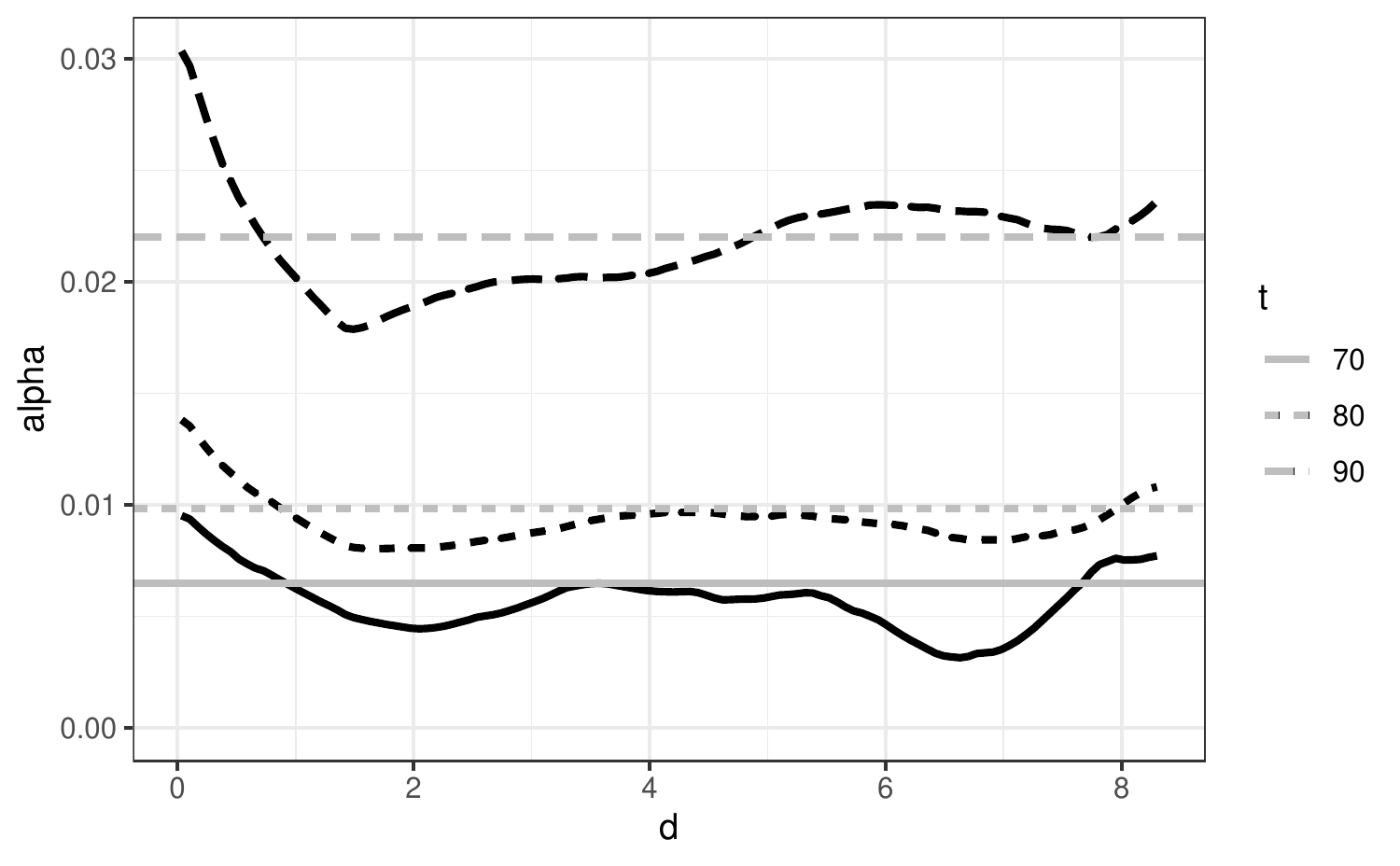}
	\caption[]{Left: Female intensity of transition from dementia to death  $\hat\alpha(t,d)$ (see \eqref{unrestest}), Right: Cross-sections of left panel at ages $t=70,80,90$ (including Markovian estimates \eqref{hjort1994}, grey lines)}
	\label{surface.plot}
\end{figure}

For fixed $d$, we see exponential growth in $t$-direction. In $d$-direction, the estimate is large at the beginning, decreases rapidly and grows later on. However, for large $d$, the trend per $t$ looks different, even though it may also be a small-sample impression, because observations thin out considerably in that area. The pattern is also evident for three cross-sections in Figure \ref{surface.plot} (right panel) for fixing $t=70, 80$ and $90$ years. The cross-sections are also compared to the estimator \eqref{hjort1994} of $\alpha(t,d)$, under the Markov hypothesis. i.e. constant in $d$-direction. We now apply the test in Theorem \ref{Gozalo_2.th}(a), by rejecting for $S_n(\hat X)^2$ being too large. Note that for the variance $\sigma_x^2$ and its estimator $\hat{\sigma}_x^2$ from Theorem \ref{satz3_1}(a)+(c), all women now enter the analysis, but due to multiplication by $\sqrt{n}$, again only the women with dementia become relevant. For a closer look at such truncation aspects, also of women deceased before 2004, see \cite{weissbachm2021effect}.

\begin{figure}[htb] 
	\centering
	\includegraphics[height=0.4\textheight]{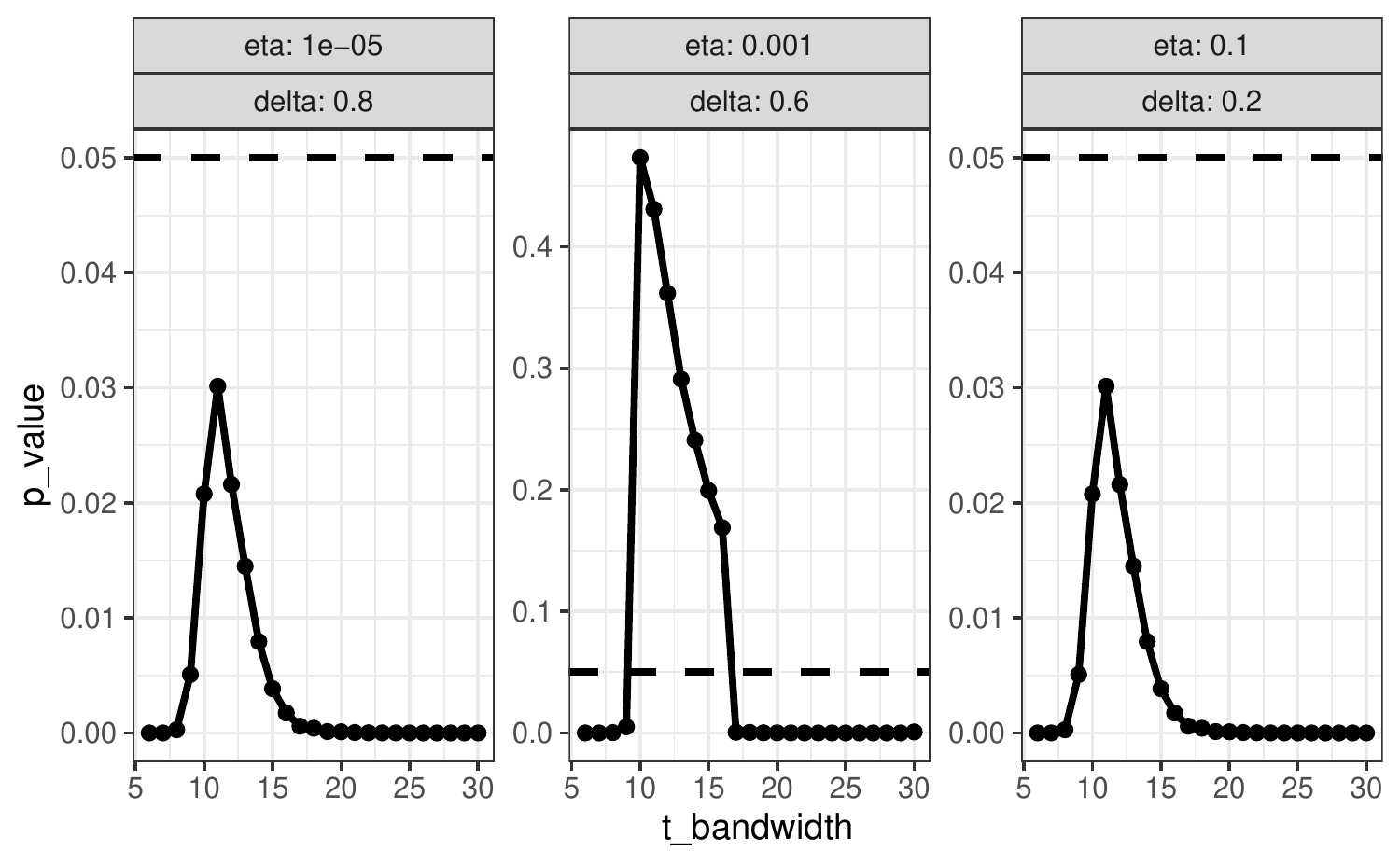}
	\caption[]{P-values dependent of bandwidths in $t$-direction  for tests at 5\% level (indicated as line).(Bandwidth in $d$-direction is 65\% thereof.)}
	\label{random_test.plot}
\end{figure}

We select $\zeta=15$ grid points. The grid is determined by drawing from a uniform distribution on $[70, 90]\times[2,4]$. In this interval, sufficiently many observations allow an estimation. Together with the bandwidth, hyper-parameters  $\eta$ and $\delta$ need to be chosen. Figure \ref{random_test.plot} exhibits the resulting p-values as a function of the bandwidths in $t$-direction for three choices. The same grid is used for all tests, and for newly drawn grids for each parameter combination, the plots are less smooth. The p-values are uniformly below the 5\%-level for many choices of $\eta$ and $\delta$, but can occasionally exceed the level for some, even up to $0.4$. All in all, still some doubt remains about the significance of the death intensity of those suffering from dementia on the disease duration.

%\newpage
\section{Discussion}

One can extend from one covariate, here representing the state duration, to several. We can test more hypotheses like that of a Cox regression. Also, one could consider testing the hypothesis of Markovianity by simultaneously testing duration dependence for all possible state combinations. For such a test, the asymptotic correlation between estimators for different combinations needs to be investigated. In the case of independence, the test statistics will then be distributed with the number of combinations as degrees of freedom.

%The application is only for chronic disease, hence with progressive history. For a non-chronic disease, the semi-Markov models is applicable and covered by us.    

%The paper is an indication, that without kernel smoothing, properties like to Markovian cannot be assessed, in contrast to goodness-of-fit testing in the Kolmogorov-Smirnov-type testing \cite[see e.g.][]{kolmogorov:2006}.  

Strictly speaking, we do not allow for right-censoring in the case of a semi-Markov process. However, one of its states can be defined as a censoring state, so that - with slightly stronger assumptions - censoring can be accounted for.

The data are independent in the cross-section, but account for a longitudinal dependency that is stronger than  Markovian. Hence, asymptotic analysis could use the aforementioned independence. However, it is interesting to note that to a certain degree, the results rely on Theorem 2.3 from \cite{Gozalo.1993}, which itself relies on a martingale limit theorem of \cite{Bierens.1984}. This author analyses a time series model, i.e. without being able to exploit cross-sectional independence and relies on the same result from the 1960s from Robert Jennrich. So do \cite{weisradl2019} who again use the cross-sectional independence in a panel model.

\textbf{Acknowledgment}:
The financial support from the Deutsche Forschungsgemeinschaft (DFG)
for R. Wei\ss bach and G. Doblhammer is gratefully acknowledged (Project 386913674  `Multi-state, multi-time, multi-level analysis of health-related demographic events: Statistical aspects and applications'). We thank M. Trede and P.K. Andersen for their valuable input in discussions. For the generous support with data, we thank the AOK Research Institute (WIdO). The linguistic and idiomatic advice of Brian Bloch is also gratefully acknowledged. Some results of this article are contained in the dissertation of L. Radloff at the Department of Economics, Universit\"at Rostock.

%\newpage

%\newpage
\begin{appendix}

\section{Proofs of intermediate results} \label{furthres}

\subsection{Minor Lemmas and formula}

\begin{Lem} \label{Bedingungstausch.lem}
	Let $X,Y$ be random variables on the probability space $(\Omega, \mathfrak{A},P)$ and $Z\in L_1(\Omega, \mathfrak{A},P)$. Let further $A\in\sigma(X)\cap \sigma(Y)\subset \mathfrak{A}$ with $P(A)>0$.\\If $\sigma(X|_A)=\sigma(Y|_A)~\big(\subset \mathfrak{A}|_A\big)$, then also $ \IE[Z|X]\big|_A=\IE[Z|Y]\big|_A, \hspace{0.5cm} P_A\text{-a.s.}.$
\end{Lem}

\begin{proof}

We start by proving the following \\
\textit{Assertion (*):} From $X:(\Omega, \mathfrak{A})\rightarrow (\Xi, \mathfrak{B})$ measurable and $A\in\sigma(X)$ follows
$ \sigma(X|_A) = \sigma(X)\cap \mathfrak{P}(A). $
This means that, for $A\in\sigma(X)$, the $\sigma$-algebra generated by $X|_A$ is composed of those sets of $\sigma(X)$ which are subsets of $A$.\\
\textit{Proof of (*):} ``$\subset$'': By $C\in\sigma(X|_A)$ there is a set $B\in\mathfrak{B}$ such that $C = \{\omega\in A| X(\omega) \in B\} = \{\omega \in \Omega| X(\omega)\in B\} \cap A.$
Since both sets of that intersection are in $\sigma(X)$, the same is true for $C$.\\
``$\supset$'': By $C \in \sigma(X)\cap \mathfrak{P}(A)$ there is a set $B\in\mathfrak{B}$ such that $C=\{\omega \in \Omega| X(\omega)\in B\}$ and $C\subset A$. This already implies $C = \{\omega \in A| X(\omega)\in B\} \in \sigma(X|_A),$ which proves Assertion (*).\\
We will now utilize that two functions with identical $\mu$-integrals over all measurable sets are $\mu$-a.e. identical \cite[see e.g.][chapter IV, Theorem 4.4]{Elstrodt.2009}. Hence, it is sufficient to prove that for all $B\in\sigma(X|_A)$ we have $\int_B \IE[Z|X]\big|_A dP_A = \int_B \IE[Z|Y]\big|_A dP_A .$
In that case the Lemma's statement would be shown $P_A|_{\sigma(X|_A)}$-a.s. and thus $P_A$-a.s.. We obtain
\begin{align*}
	\int_B \IE[Z|X]\big|_A dP_A \stackrel{(i)}{=} \int_B \IE[Z|X] dP \stackrel{(ii)}{=} \int_B Z dP \stackrel{(ii)}{=} \int_B \IE[Z|Y] dP \stackrel{(i)}{=}\int_B \IE[Z|Y]\big|_A dP_A. 
\end{align*}
For the first and last indentity (i) $A\in\mathfrak{A}$, $B\in\mathfrak{A}$ and $B\subset A$ guarantee that both sides are well defined. Equality follows since $P_A=P$ on its domain. The other identities (ii) are true according to the definition of conditional expectations: Assertion (*) provides both $B\in\sigma(X)$ and $B\in\sigma(Y)$ due to $B\in\sigma(X|_A)$ and $A\in\sigma(X)$.
\end{proof}

A short calculation yields:
\begin{Lem} \label{y_messbar.lem}
	Let $\mathcal{S}$ and $\mathcal{T}$ be index sets. Further let $(\Omega, \mathfrak{A})$, $(\mathcal{X}_s, \mathfrak{B}_s), s\in\mathcal{S},$ and $(\mathcal{Y}_t, \mathfrak{C}_t), t\in\mathcal{T},$ be measurable spaces. Let $X_s:\Omega\rightarrow\mathcal{X}_s, s\in\mathcal{S},$ and $Y_t:\Omega\rightarrow\mathcal{Y}_t, t\in\mathcal{T},$ be measurable mappings. 
	Assume that for all $s\in\mathcal{S}$ there is a subset $\mathcal{T}_s \subset \mathcal{T}$ and a $(\bigotimes_{t\in\mathcal{T}_s}\mathfrak{C}_t)$-$\mathfrak{B}_s$-measurable mapping $g:\bigtimes_{t\in\mathcal{T}_s} \mathcal{Y}_t\rightarrow \mathcal{X}_s$ with $X_s=g( \{Y_t\}_{t\in\mathcal{T}_s} ). $
	Then we have $\sigma(X_s, s\in\mathcal{S}) \subset \sigma(Y_t, t\in\mathcal{T}).$ 
\end{Lem}

\begin{Lem} \label{uniform_int.lem}
	Let $A\subset\IR^2$ and $Q\subset\IR^2$ be compact and $g:\IR^2\rightarrow \IR$ continuous. Let $b_n$ be a sequence approaching zero. Then $\sup_{x\in A} \left|\int_Q \left\{ g(x) - g(x+b_nq) \right\}dq \right| \rightarrow 0$.
\end{Lem}
\begin{proof} With $C=\{x+bq|x\in A, q\in Q, b \in[0,1]\}$ $g$ is uniformly continuous on $C$, since this set is compact. That means, for each $\epsilon>0$ there is $\delta>0$, such that for all $x,y\in C$ we have $|x-y|<\delta ~\Rightarrow~|g(x)-g(y)|<\epsilon$, where $|\cdot|$ denotes euclidean norm on $\IR^2$, too.  Let $\epsilon>0$ and $\delta>0$ according to this definition. Because $Q$ is compact, there is $N$, such that for all $n\ge N$, $b_n \le 1$ and $|b_n q| < \delta~ \forall q\in Q$. For all $n\ge N$ then
	\begin{align*}
		\sup_{x\in A} \left|\int_Q g(x) - g(x+b_nq) dq \right| \le \sup_{x\in A} \int_Q \epsilon dq = \epsilon  \lambda_2(Q),
	\end{align*}
	where $\lambda_2$ is the $2$-dimensional Lebesgue-measure. Since $\lambda_2(Q)<\infty$, the proof is complete. 
\end{proof}

For the proof of formula \eqref{alpha_def.leam}, we first note that
\begin{align} \label{set_in_salg.eq}
\{\mu(t) = m\}&\in\sigma(X(u),u\le t) ~\text{and} \nonumber\\
\{\mu(t) = m\}&\in\sigma(S_0, Z_0, ...,S_m, Z_m, \I_{\{Z_{m+1}>t\}}).
\end{align} 
Further, we recognize
\begin{align} \label{sigma_alg_id.eq}
\sigma((X(u), u\le t)|_{\{\mu(t) = m\}}) = \sigma((S_0, Z_0, ..., S_m, Z_m, \I_{\{Z_{m+1}>t\}})|_{\{\mu(t) = m\}}),
\end{align} 
because for all $\omega\in\{\mu(t) = m\}$ we have $X(u,\omega) = \sum_{j=0}^m\I_{[Z_m(\omega), Z_{m+1}(\omega))} S_j(\omega)$ on the one hand, and on the other
$Z_0(\omega) = 0$, $S_0(\omega) = X(0,\omega)$,
\[
Z_{j+1}(\omega) = \inf\left\{ \left. q \I_{\{q>Z_j(\omega), X(q, \omega) \ne S_j(\omega)\}} + \infty (\I_{\{q\le Z_j(\omega)\}} + \I_{\{X(q, \omega) = S_j(\omega)\}}) \right| q\in\IQ \right\} 
\]
for $j=0,...,m-1$, $S_{j+1}(\omega) = X(Z_{m+1}(\omega))$ (for $j=0,...,m-1$) and 
$\I_{\{Z_{m+1}>t\}}(\omega)\equiv 1$. Lemma \ref{y_messbar.lem} then yields (\ref{sigma_alg_id.eq}). We utilize this for a transformation:
\begin{align}
P&(X(t+h) = s| X(t) = r , \tilde D(t) = d, X(u), u\le t)|_{\{\mu(t)=m\}}\nonumber\\
&\stackrel{(i)}{=}P(X(t+h) = s| S_0, Z_0, ...,S_{m-1}, Z_{m-1}, S_m = r , Z_m = t-d, T_{m+1}>d)|_{\{\mu(t)=m\}}\nonumber\\	
&\stackrel{(ii)}{=}  \sum_{k=1}^\infty P( \sum_{j=1}^kT_{m+j} \le d+h, \sum_{j=1}^{k+1}T_{m+j} > d+h, S_{m+k} = s |  S_m = r, Z_m = t-d, T_{m+1}>d ) \label{sm_intensity.eq}
\end{align}
(\ref{sigma_alg_id.eq}) and (\ref{set_in_salg.eq}) allow application of Lemma \ref{Bedingungstausch.lem} to obtain (i). For (ii) we form a countable partition of the events $\{X(t+h) = s, \mu(t) = m, S_m = r\}$ and exploit $\sigma$-additivity. For these smaller events the definition of a semi-Markov process makes clear that $S_0, Z_0, ..., S_{m-1}, Z_{m-1}$ can be removed from the condition. Since the new term is independent of $\omega$, we can omit restriction of the domain at this point. \\
We start by examining the first summand of this series. The condition $T_{m+1}>d$ can be handled according to rules for calculation with elementary conditional probabilities. Thus, we get a fraction with denominator
\begin{align} 
P&(T_{m+1}>d|S_m=r, Z_m = t-d) = S_{t-d}^r(d) \label{denom_first.eq}
\intertext{and numerator}
P&(d<T_{m+1}\le d+h < T_{m+1}+T_{m+2}, S_{m+1} = s | S_m = r, Z_m = t-d) \nonumber\\
&=\int_d^{d+h} \int_{d+h-t_1}^\infty f_{t-d}^{rs}(t_1) f_{t-d+t_1}^s(t_2) dt_2 dt_1 \nonumber\\
&=\int_d^{d+h} \int_{0}^\infty f_{t-d}^{rs}(t_1) f_{t-d+t_1}^s(t_2) dt_2 dt_1 - \int_d^{d+h} \int_{0}^{d+h-t_1} f_{t-d}^{rs}(t_1) f_{t-d+t_1}^s(t_2) dt_2 dt_1 \nonumber\\
&=\int_d^{d+h} f_{t-d}^{rs}(t_1) dt_1 + O(h^2). \label{num_first.eq}
\end{align}
The second density in the first summand of the third line integrates to one. The second summand's integrand is bounded due to (L). These explain the last identity. The Landau-$O$-notation here, as in the following, refers to the limit $h\rightarrow0$. \\
The $k$th summand can be bounded by
\begin{align*}
P(\mu(t+h) - \mu(t) \ge k | S_m = r, Z_m = t-d, T_{m+1} > d ) \le C^k h^k,
\end{align*} 
according to Lemma \ref{dMu.lem}. This establishes that the sum in (\ref{sm_intensity.eq}), excluding the first summand, is $O(h^2)$ because $\sum_{k=2}^\infty C^kh^k = (C^2 h^2)/(1-Ch) = O(h^2)$ for $h$ sufficiently small.
That, combined with (\ref{denom_first.eq}) and (\ref{num_first.eq}), plugged into (\ref{sm_intensity.eq}), gives
\begin{align*}
\frac 1h 	P&(X(t+h) = s| X(t) = r , \tilde D(t) = d, X(u), u\le t)|_{\{\mu(t)=m\}} \\
&= \frac 1h \left\{\frac{\int_d^{d+h} f_{t-d}^{rs}(t_1) dt_1 + O(h^2)}{S_{t-d}^r(d)} + O(h^2) \right\}
 \underset{h\rightarrow0}\longrightarrow \frac{f_{t-d}^{rs}(d)}{S_{t-d}^r(d)}.
\end{align*}
Since this is true for all $m\in\IN$, and $\{\mu(t)=m\}$ form a partition of $\Omega$, formula \eqref{alpha_def.leam} is proved.

\subsection{Proof of Theorem \ref{fdd_conv.th}} 

Lemma \ref{mvmlt.lem} shall be applied with
$H_{i,j}^{(n)}(t) = n^{-1/2}b^{(d+1)/2}K_b(x_j - X_i(t))$, $j=1,...,\zeta, i=1,...,n$.
Furthermore, we notice $d\langle M_i\rangle (s)= \alpha(s, Z_i(s))Y_i(s)ds$ \cite[see e.g.][S. 74, formula (2.4.3)]{Andersen.1993}. We aim to verify (G1') by first proving the convergence of expected values and subsequently the asymptotic vanishing of variances. We have
\begin{align*}
	\IE &\left[ \sum_{i=1}^n \int H_{i,j}^{(n)}(s)H_{i,k}^{(n)}(s) d\langle M_i \rangle (s) \right]\\
	&\stackrel{(i)}{=}b^{d+1} \int_{[0,1]^{d+1}}K_b(x_l - w)K_b(x_j-w)\alpha(w)\varphi(w)dw\\
	&\stackrel{(ii)}{=} \int_{-x_j/b}^{(1-x_j)/b} K(q) K\left( \frac{x_j - x_l}b +q \right) \alpha(x_j + bq)\varphi(x_j + bq) dq \\
	&\stackrel{(iii)}{\longrightarrow} 
	\begin{cases}
		\kappa_2^{d+1}\alpha(x_l)\varphi(x_l) & l = j \\
		0 & l \ne j.
	\end{cases}
\end{align*}
Due to $d\langle M_i\rangle (s)= \alpha(s, Z_i(s))Y_i(s)ds$ (i) is provided by \eqref{EV.lem}. We then substitute $q = (x_j - w) / b$ and use symmetry of $K$ for (ii). The case $l=j$ in (iii) is already covered by \cite{Linton.1995}. For $l \ne j$ we first note that for sufficiently large $n$, hence, small $b$, the integration are contains the support of $K$ and we thus can integrate over $[-1, 1]^2$ instead. For sufficiently small $b$ we have $K\left( (x_l - x_j)/b +q \right) = 0 $ for all $q\in[-1,1]^2$. Therefore the integrand and, hence, the integral, equals zero.
For the variance we obtain
\begin{align*}
	Var&\left[ \sum_{i=1}^n \int H_{i,j}^{(n)}(s)H_{i,k}^{(n)}(s) d\langle M_i \rangle (s) \right]\\
	&\stackrel{(i)}{\le}\frac 1n b^{2d+2}\int_{[0,1]^{d+1}}K_b^2(x_l - w)K_b^2(x_j-w)\alpha^2(w)\varphi(w)dw\\
	&\stackrel{(ii)}{=} \frac 1n b^{2d+2 -4d - 4 + d+1} \int_{-x_j/b}^{(1-x_j)/b} K^2(q)K^2\left( \frac{x_j-x_l}{b}+q \right) \alpha^2(x_j+bq)\varphi(x_j+bq) dq\\
	&\stackrel{(iii)}{=} o(1).
\end{align*}
Formula (10) from \cite{Linton.1995} provides (i). We then substitute $q = (x_j - w) / b$ and utilize symmetry of $K$ (by assumption (K)) for (ii). Again, for sufficiently large $n$ the support of $K$ is contained in the integration area and we can integrate over $[0,1]^2$. Since all integrands are bounded on this set, the same is true for the integral as a whole.  The prefactor is $o(1)$ according to assumption (B). Therefore, (iii) is clear. For $j\ne l$ the integral even equals zero for sufficiently large $n$ for the same reasons as in the calculation of the expected values.

\noindent (G2') is already covered by \cite{Linton.1995}, as well as $\mathscr C_x \rightarrow_P \varphi(x)$ for all $x_j, j=1,...,\zeta$.

\subsection{Proof of Lemma \ref{gozalo_problem.lem} (condensed version)} 

It suffices to prove the statement when substituting $S_n$ with the first summand of the following decomposition, because asymptotically, under $H_0$, $\hat\beta$ is closer to $\alpha$ than $\hat\alpha$. 
\begin{align*}
	S_n(\tilde X) = (nb^{d+1})^{\frac12} \Bigg\{ \left( \frac{\hat\alpha(\tilde X) - \alpha(\tilde X)}{\hat\sigma_{\tilde X}} \right) +  \left( \frac{\alpha(\tilde X) - \hat\beta(\tilde X)}{\hat\sigma_{\tilde X}} \right) \Bigg\}
\end{align*} 
We denote this summand with $S_n'$ and can represent it as follows (see \eqref{alpha_decomp.eq}).
\begin{align*}
	S_n'(\tilde X):&=(nb^{d+1})^{\frac12} \left( \frac{\hat\alpha(\tilde X) - \alpha(\tilde X)}{\hat\sigma_{\tilde X}} \right) \\&= \frac 1 {\hat\sigma_{\tilde X}} \frac 1{\mathscr{C}_{\tilde X}} (nb^{d+1})^{\frac12} \Bigg\{ \frac 1n \sum_{i=1}^n \int K_b(\tilde X-X_i(s))dM_i(s) \\&~~+ \frac 1n \sum_{i=1}^n \int K_b(\tilde X-X_i(s))\left[ \alpha(X_i(s)) - \alpha(\hat X) \right]Y_i(s)ds\Bigg\}\\&=: \frac 1 {\hat\sigma_{\tilde X}} \frac 1{\mathscr{C}_{\tilde X}} \left\{  R_{n1}(\tilde X) + R_{n2}(\tilde X)\right\} =:\frac 1 {\hat\sigma_{\tilde X}} \frac 1{\mathscr{C}_{\tilde X}} R_{n}(\tilde X)
\end{align*}
The remainder of the proof is structured by five assertions that, combined, establish the lemma. Here, we will only point to the most important ideas for their respective proofs.\\
Firstly, uniform convergence of $\hat\sigma_x$ to $\sigma_x$ and $\mathscr C_x$ to $\varphi(x)$, together with the fact that both $\sigma_x$ and $\varphi(x)$ are bounded away from zero on the compact $\mathscr X$, gives 
\textit{Assertion 1: } If $1/\zeta_n \cdot \sum_{j=1}^{\zeta_n} R_n(\tilde X^j)^2 = O_P(1)$, then $1/\zeta_n \cdot \sum_{j=1}^{\zeta_n} S_n'(\tilde X^j)^2 = O_P(1)$, too.\\
Secondly, we have
\textit{Assertion 2: } If $\sup_x \IE R_{nl}(x)^4 = O(1), l=1,2$, then $1/\zeta_n \cdot \sum_{j=1}^{\zeta_n} R_n(\tilde X^j)^2 = O_P(1)$, too.\\
For proof one first shows that $\IE R_{nl}(\tilde X)^4 = O(1), l=1,2$, too, by conditioning on $\tilde X = x$ and exploiting the independence between $\tilde X$ and the data. Using Minkowski inequality $\IE R_n(\tilde X)^4=O(1)$ follows. Now, expectations and variances of $1/\zeta_n \cdot \sum_{j=1}^{\zeta_n} R_n(\tilde X^j)^2$ stay bounded and Chebychev's inequality delivers the assertion.\\
As each of the random variables $R_{nl}(x)$ is of the form $1/n\sum_{i=1}^n H_{ni}$ for certain random variables $H_{in}$, the third assertion will investigate this particular structure.\\
\textit{Assertion 3: } For $n\in\IN$ let $H_{ni}, i=1,...,n$, be i.i.d. random variables. If $\IE H_{ni} = O((nb^{d+1})^{-1/2})$, $VarH_{n1} = O(b^{-d-1})$ and $\IE H_{n1}^4 = O(b^{-3d-3})$, then also $$\IE \left[(nb^{d+1})^{1/2} 1/n \sum_{i=1}^n H_{ni} \right]^4 = O(1).$$ If instead the summands are stochastic processes dependent on some parameter $x$, and the prerequisites are valid uniformly in $x$, then the conclusion is valid uniformly in $x$, too. Uniformity here means that suprema of expected values are considered, not expected values of suprema.\\
The proof starts by ascertaining, using Hölder inequality, that we can instead examine $\IE \left[(nb^{d+1})^{1/2} 1/n \sum_{i=1}^n (H_{ni} -\IE H_{ni}) \right]^4 $ if the sequence $\IE H_{n1}$ is bounded. This term can be described in terms of multiples of $\IE\left[ H_{n1}-\IE H_{n1} \right]^4$ and $Var[H_{n1}]^2$, which are suitably bounded by assumption.\\
Finally, the random variables $R_{nl}, l=1,2$, are examined individually in order to complete the proof.\\
\textit{Assertion 4: } $\sup_x \IE R_{n2}(x)^4 = O(1)$.\\
Using notation of Assertion 3, choose $H_{ni}(x)=\int K_b(x-X_i(s))\left[ \alpha(X_i(s)) - \alpha(x) \right]Y_i(s)ds$. From the proof of Theorem 1 b) in \cite{Linton.1995} it is easy to see that $\sup_x \IE H_{ni}(x) = O(b^2)$ and $\sup_x Var H_{ni}(x) = O(b^{-d-1})$ which, using (\~ B), also implies $\sup_x \IE H_ni(x) = O((nb^{d+1})^{-1/2})$. It is left to show $\sup_x \IE H_{ni}(x)^4 = O(b^{-3d-3})$ which can be done using Hölder inequality and \eqref{EV.lem}.\\
\textit{Assertion 5: } $\sup_x \IE R_{n1}(x)^4 = O(1)$.\\
$R_{n1}(x)=\frac 1n \sum_{i=1}^n \int K_b(x-X_i(s))dM_i(s)$ is a martingale for each $x$. Hence the expected value equals zero for all $x$. The bound $\sup_x Var \int K_b(x-X_i(s))dM_i(s)=O(b^{-d-1})$ is found easily. For the fourth moments a decomposition is used:
\begin{align*}
	\int K_b(x-X_i(s))dM_i(s) = \int K_b(x-X_i(s))dN_i(s) - \int K_b(x-X_i(s))\alpha(X_i(s))Y_i(s)ds \\\le \int K_b(x-X_i(s))dN_i(s) + \int K_b(x-X_i(s))\alpha(X_i(s))Y_i(s)ds
\end{align*}
It suffices if both summands fulfil the assumptions of Assertion 3 individually. For the second summand argumentation follows that from Assertion 4. For the first summand we substitute the counting process $N$ with a homogeneous Poisson process $\tilde N$ that has strictly greater intensity process. This is possible since $\alpha(x)$ is bounded. Then Hölder inequality provides
\begin{align*}
	\left[ \int K_b(x-X_i(s))Y_i(s)d\tilde N(s) \right]^4 \le \tilde N(1)^3 \int K_b^4(x-X_i(s))Y_i(s)d\tilde N(s)  .
\end{align*}
Next, we condition on the number of jumps of $\tilde N$, obtaining $$\IE \left[ \left. \tilde N(1)^3 \int K_b^4(x-X_i(s))Y_i(s)d\tilde N(s) \right| \tilde N(1) = k \right] = k^4 O(b^{-3d-3}).$$ Recognizing that the Poisson distribution has finite fourth moment the proof can be completed.

%\newpage
\section{Martigale properties (Theorem \ref{martingale_property.th})} 

\subsection{With censoring} \label{withcens}
Proof of right-continuity of $\mathfrak F_t$ and predictability are identical to the uncensored case. Argumentation for martingale property follows a similar path, too. Again, let $0<s<t$. Analogous to (\ref{dM.eq}) we notice, that
\begin{align} \label{dM2.eq}
	M(t) - M(s) = \I_{\{ s<Z_2 \le t, \Delta = 2 \}} - \int_s^t Y(u) \alpha_{12}(u, D(u)) du.
\end{align}	
Analogous to (\ref{EdM_comp.eq}) we examine the conditional expectation on three different events separately:
\begin{equation*}
	\begin{split}
	\IE[M(t) - M(s)|\mathfrak{F}_s] =  \I_{\{Z_2 \le s\}}\IE[M(t) - M(s)|\mathfrak{F}_s] \Big| _{\{Z_2 \le s\}}\\
	+  \I_{\{Z_1>s\}} \IE[M(t) - M(s)|\mathfrak{F}_s] \Big|_{\{Z_1>s\}}  
	+  \I_{\{Z_1 \le s < Z_2\}} \IE[M(t) - M(s)|\mathfrak{F}_s] \Big|_{\{Z_1 \le s < Z_2\}}
	\end{split}
\end{equation*}
On the set $\{Z_2\le s\}$ we obviously have $\I_{\{ s<Z_2 \le t, \Delta = 2 \}} = 0$. Since on this event $Y(t) \equiv 0$, too, it follows directly by (\ref{dM2.eq}) that $\I_{\{Z_2 \le s\}}\IE[M(t) - M(s)|\mathfrak{F}_s] \Big| _{\{Z_2 \le s\}} = 0$.
On the event $\{ Z_1 > s \}$ conditional expectations must be constant, following the identical argument as in the uncensored case. For $c:=\IE[\I_{\{ s<Z_2 \le t, \Delta = 2 \}}|\mathfrak{F}_s] \Big|_{\{Z_1>s\}}$ we obtain according to the definition of conditional expectations
\begin{align*}
	c P(Z_1 > s) = \int_{\{ Z_1 > s \}} cdP = \int_{\{ Z_1 > s \}} \I_{\{ s<Z_2 \le t, \Delta = 2 \}}dP\\ = P(s<Z_1<Z_2\le t, T_1+T_2\le U). 
\end{align*}
Further, for  $\tilde c:=\IE\left[ \int_s^t Y(v) \alpha_{12}(v, D(v)) dv|\mathfrak{F}_s \right] \Big|_{\{Z_1>s\}}$ we have
\begin{align*}
	\tilde c  &P(Z_1 > s) = \int_{\{ Z_1 > s \}} \tilde cdP = \int_{\{ Z_1 > s \}} \int_s^t Y(v) \alpha_{12}(v, D(v)) dvdP\\
	&\stackrel{(i)}{=}\IE_{T_1} \left[ \IE \left[ \int_s^t \I_{\{ s<T_1<v\le Z_2\}} \alpha_{12}(v, v-T_1) dv \Big| T_1 = t_1\right]  \right]\\ 
	&\stackrel{(ii)}{=} \IE_{T_1} \left[ \int_s^t \alpha(v, v-t_1)  P \left( (T_1+T_2)\wedge U \ge v\Big| T_1 =t_1 \right) \I_{(s,v)}(t_1) dv \right]\\
	&\stackrel{(iii)}{=}\IE_{T_1} \left[ \int_s^t \frac{f_{T_2|T_1=t_1}(v-t_1)}{S_{T_2|T_1=t_1}(v-t_1)} S_{T_2|T_1=t_1}(v-t_1)P(U\ge v|T_1 = t_1) \I_{(s,v)}(t_1) dv \right] \\
	&\stackrel{(iv)}{=}\IE_{T_1} \left[ \I_{(s,\infty)}(t_1) \int_{t_1}^t f_{T_2|T_1}(v - t_1) P(U\ge v|T_1 = t_1)  dv \right]\\
	&\stackrel{(v)}{=}\IE_{T_1} \left[ \I_{(s,\infty)}(t_1) P(T_1+T_2 \le t, U\ge T_1+T_2|T_1 = t_1)  \right]\\
	&\stackrel{(vi)}{=}P(s<T_1<T_1+T_2\le t, U\ge T_1+T_2) \\
	&= P(s<Z_1<Z_2\le t, T_1+T_2\le U)
\end{align*}
For (i) we first combine $Y(v)$ with the integration area to a single indicator function. Then the outer integral is an expexted value, whereby we also add a conditional expectation. Furthermore we have $D(v) = v - T_1$  on the event $\{Y(v) \ne 0 \}$. Next, we change order of integration and make use of multiplication theorem for conditional expectations in order to extract the $T_1$-measurable factors from the conditional expectation. This provides (ii). (iii) follows from \eqref{intensities_through_densities.lem} and  
\begin{align} \label{TminU.eq}
	\begin{aligned}
		P &\left( (T_1+T_2)\wedge U \ge v\Big| T_1 =t_1 \right) = P \left( T_1+T_2 \ge v,  U \ge v\Big| T_1 =t_1 \right) \\
		&=P \left( T_1+T_2 \ge v\Big| T_1 =t_1 \right)P \left( U \ge v \right)\\
		&=S_{T_2|T_1=t_1}(v-t_1) P \left( U \ge v \big| T_1=t_1\right).
	\end{aligned}
\end{align}
This factorization is valid according to \cite[Theorem 1.120]{Witting.1985}, because the random vectors
	$$ \begin{pmatrix}
		\I_{\{T_1 + T_2 \ge v\}} \\ T_1
	\end{pmatrix} 
	\text{ and } 
	\begin{pmatrix}
		\I_{\{U \ge v\}} \\ 1
	\end{pmatrix} $$ 
are stochastically independent by assumption. Also the calculation (\ref{cond_surv.eq}) has to be applied again. We obtain (iv) by reducing the fraction plus $\I_{(s,v)}(t_1) = \I_{(s,\infty)}(t_1) \I_{(t_1, \infty)}(v)$. The second indicator is moved to the integral limit. For (v) we first notice that $T_2$ and $U$ are stochastically independent also dependent on $T_1 = t_1$. This can be shown with similar arguments as in step (iii) of this enumeration. We then obtain
\begin{align*}
	P&(T_1+T_2 \le t, U\ge T_1+T_2|T_1 = t_1) \\
	&= \int \int \I_{(0,t]}(t_1 + t_2) \I_{[t_1 + t_2,\infty)}(u) dP_{(U,T_2)|T_1=t_1}(u,t_2) \\
	&= \int \I_{(-t_1, t-t_1]}(t_2) \left[ \int \I_{[t_1 + t_2, \infty)}(u) dP_{U|T_1 = t_1}(u) \right] dP_{T_2|T_1=t_1}(t_2) \\
	&= \int_0^{t - t_1} P(U\ge t_1+ t_2 | T_1 = t_1) f_{T_2|T_1 = t_1}(t_2) dt_2\\
	&= \int_{t_1}^{t} P(U\ge v | T_1 = t_1) f_{T_2|T_1 = t_1}(v - t_1) dv.
\end{align*}
Here again Theorems 1.122 and 1.126 from \cite{Witting.1985} are applied. The third identity exploits that $T_2$ is an a.s. positive random variable. Finally, the calculation of the outer expected value in (vi) is explained by the intermediate steps
\begin{align*}
	\IE_{T_1} &\left[ \I_{(s,\infty)}(t_1) P(T_1+T_2 \le t, U\ge T_1+T_2|T_1 = t_1)  \right]\\ 
	&= \IE \left[ \I_{\{T_1 > s\}} \IE\left( \I_{\{T_1+T_2 \le t, U\ge T_1+T_2 \}} \big|T_1 \right) \right] = \IE \left[ \IE\left( \I_{\{T_1 > s, T_1+T_2 \le t, U\ge T_1+T_2 \}} \big|T_1 \right) \right]\\
	&= P(s < T_1 < T_1 + T_2 \le t, U \ge T_1 + T_2).
\end{align*}
It remains to examine the conditional expectation on the set $\{Z_1 \le s < Z_2\}$. By reasons analogue to the uncensored case $\IE[\cdot| \mathfrak{F}_s]\big|_{\{Z_1 \le s < Z_2\}} = \IE[\cdot| T_1, \I_{\{Z_2 > s\}}\big|_{\{Z_1 \le s < Z_2\}}$. 
We first consider
\begin{align*}
	\IE \big[ & \I_{\{s < Z_2 \le t, \Delta = 2\}} |T_1 = t_1 \big] = \int \int \I_{(s,t]}(t_1 + t_2) \I_{[t_1+t_2, \infty)}(u) dP_{(U, T_2)|T_1 = t_1}(u, t_2) \\
	&=\int \I_{(s-t_1, t-t_1]}(t_2) \left[ \int \I_{[t_1 + t_2, \infty)}(u) dP_{U|T_1 = t_1}(u) \right] dP_{T_2|T_1=t_1}(t_2) \\
	&= \int_{s-t_1}^{t - t_1} P(U\ge t_1+ t_2 | T_1 = t_1) f_{T_2|T_1 = t_1}(t_2) dt_2\\
	&= \int_{s \vee t_1}^{t} P(U\ge v | T_1 = t_1) f_{T_2|T_1 = t_1}(v - t_1) dv
\end{align*}
where we again apply Theorems 1.122 and 1.126 from \cite{Witting.1985} as well as the fact that $T_2$ is a.s. positive. We further obtain 
\begin{align*}
	\IE &\left[ \int_s^t Y(v) \alpha_{12}(v, D(v)) dv \Big| T_1 = t_1 \right] \stackrel{(i)}{=}  \int_s^t \alpha_{12}(v, v - t_1) \IE\big[ \I_{\{T_1 < v \le Z_2\}} | T_1 = t_1\big] dv \\
	&\stackrel{(ii)}{=} \int_s^t \I_{(t_1,\infty)}(v) \frac{f_{T_2|T_1=t_1}(v - t_1)}{S_{T_2|T_1=t_1}(v - t_1)} P\left( (T_1 + T_2) \wedge U \ge v | T_1 = t_1 \right) dv\\
	&\stackrel{(iii)}{=} \int_{s \vee t_1}^t \frac{f_{T_2|T_1=t_1}(v - t_1)}{S_{T_2|T_1=t_1}(v - t_1)}S_{T_2|T_1=t_1}(v - t_1) P(U \ge v | T_1 = t_1) dv\\
	&\stackrel{}{=} \int_{s \vee t_1}^t f_{T_2|T_1=t_1}(v - t_1)  P(U \ge v | T_1 = t_1) dv
\end{align*}
Changing order of integration and moving the $T_1$-measurable factor $\alpha_{12}$ out of the conditional expectation gives (i). Extracting another $T_1$-measurable factor and \eqref{intensities_through_densities.lem} provides (ii). For (iii) see \eqref{TminU.eq} together with \eqref{cond_surv.eq}.\\
For reasons analoguous to the uncensored case we have
\begin{align*}
	\IE \big[ \I_{\{s < Z_2 \le t, \Delta = 2\}} |T_1 = t_1, \I_{\{Z_2>s\}} = 1\big] &= \frac{\IE \big[ \I_{\{s < Z_2 \le t, \Delta = 2\}} |T_1 = t_1\big]}{P(Z_2>s|T_1=t_1)}
	\intertext{and}
	\IE \left[ \int_s^t Y(v) \alpha_{12}(v, D(v)) dv \Big| T_1 = t_1, \I_{\{Z_2>s\}} = 1 \right] &= \frac{\IE \left[ \int_s^t Y(v) \alpha_{12}(v, D(v)) dv \Big| T_1 = t_1\right]}{P(Z_2>s|T_1=t_1)}.
\end{align*}
With this we conclude $\IE \big[ M(t) - M(s) | T_1 = \cdot, \I_{\{Z_2>s\} } = \cdot \big]\Big|_{\{Z_1 \le s < Z_2\}}  = 0$ and, hence, also $\I_{\{Z_1 \le s < Z_2\}}\IE[M(t) - M(s)|\mathfrak{F}_s] \Big|_{\{Z_1 \le s < Z_2\}} = 0$. This completes the proof of martingale property for $M(t)$.

%\newpage
\subsection{For semi-Markov model} \label{semimarkmart}

Let the number of respective transitions in $\mathcal{X}$ up to time $t$ be denotes as   $N(t):=\sum_{m=1}^\infty \I_{\{Z_m\le t\}} \I_{\{S_{m-1} = r\}} \I_{\{S_{m} = s\}}$. Let  $Y(t):=\I_{\{\mathcal{X}(t-) = r\}} = \sum_{m=1}^\infty \I_{\{Z_{m-1} <t\le Z_m\}} \I_{\{S_{m-1} = r\}}$ further indicate whether, shortly before $t$, $N$ is at risk to migrate. And let further denote $D(t):=\tilde D(t)Y(t)$ the duration of $\mathcal{X}$ in state $r$, given $\mathcal{X}(t)=r$. Note that $N$ may jump several times.
We set $\alpha(t, d) :=   \alpha_{rs}(t, d)$. Again, it is to prove that these processes fit the framework of \cite{Linton.1995}. Therefore, once again we operate with the filtration $\mathfrak F_t:= (N(u), D(u+), Y(u+), u\le t).$ Including the right limits of $Y$ and $D$ respectively ensures right-continuity of the filtration and, thus, ``les conditions habituelles''. It has to be shown that $\alpha(t, D(t))Y(t)$ is the intensity process of $N(t)$. So we have to prove predictability of $\alpha Y$ (this follows directly from the continuity of $\alpha$ and left-continuity of $Y$) and particularly the martingale property of $M(t)= N(t) - \int_0^t \alpha(u, D(u))Y(u)du.$
We define a random element $\mathcal F_t:= (N(u), D(u+), Y(u+), u\le t)$ such that $\mathfrak F_t = \sigma(\mathcal F_t)$. Since obviously $D(t)$ and $Y(t)$ are $\mathcal F_t$-measurable there exist measurable mappings $d_t$ and $y_t$ such that $D(t)=d_t(\mathcal F_t)$ and $Y(t) = y_t(\mathcal F_t)$. We prove the following lemma.
\begin{Lem} \label{dM_cond_F.lem}
	Given the previous definitions we have (a) 
	\begin{multline*}
		\IE \left[\left. M(t+h) - M(t) \right| \mathcal F_t = f\right]\big|_{\{f|y_t(f)=1\}} \\
		= \IE\left[\left. M(t+h)-M(t) \right| Z_m=t-d_t(f),  T_{m+1}>d_t(f), S_m = r\right], ~~ \text{für}~P_{\mathcal F_t}\text{-f.a.} ~f,
	\end{multline*}
	for arbitrary $m\in\IN$, and
	(b)
	\begin{align*}
		\big| \IE&\left[\left. M(t+h) - M(t) \right| \mathcal F_t = f\right]\big|\Big|_{\{f|y_t(f)=0\}} \\
		&\le \sup_{s\in\IS\backslash\{r\}, z\in[0,1]}  \IE\left[\left. N(t+h)-N(t) \right| Z_m=z,  T_{m+1}>t-z, S_m = s\right] \\
		&+\sup_{s\in\IS\backslash\{r\}, z\in[0,1]}  \IE\left[\left. \int_t^{t+h}  \alpha(u,D(u))Y(u) du  \right| Z_m=z,  T_{m+1}>t-z, S_m = s\right], \\
		&~~~~~~~~~~~~~~~~~~~~~~~~~~~~~~~~~~~~~~~~~~~~~~~~~~~~~~~~~~~~~~~~~~~~~~~~~~~~~~~~~~~~\text{für}~P_{\mathcal F_t}\text{-f.a.} ~f,
	\end{align*}
	again for arbitrary $m\in\IN$.
\end{Lem}
\begin{proof}
	
(a) Subsequently we denote $\Delta M:= M(t+h) - M(t)$, $\Delta N:= N(t+h)-N(t)$ and $\Delta A:=\int_t^{t+h}  \alpha(u,D(u))Y(u) du$. We start by showing the following assertion.\\
\textit{Assertion 1: }For all $t\in[0,1]$ and every $r\in\IS$ there is a mapping $\prescript{M}{}g_t^r(z)$ independent of $m$ and $ s_0, z_0,...,s_{m-1}, z_{m-1}$ with
\begin{align*}
	\IE\left[\left. \Delta M \right| S_0=s_0, Z_0=z_0, \dots, S_m = r, Z_m=z,  T_{m+1}>t-z\right] = \prescript{M}{}g_t^r(z).
\end{align*}
Analoguously, there are mappings $\prescript{N}{}g_t^r(z)$ and $\prescript{A}{}g_t^r(z)$, for $\Delta N$ and $\Delta A$ instead of $\Delta M$.\\
\textit{Proof of Assertion 1: } $\Delta M =\Delta N - \Delta A$, therefore the verification for $\Delta N$ and $\Delta A$ suffices. First, we consider $\Delta N$. According to the definition of $N$ ist is easy to see that $\Delta N = \sum_{\nu=0}^\infty  \I_{\{S_{\nu} = r\}} \I_{\{S_{\nu+1} = s\}} \I_{\{t<Z_{\nu+1}\le t+h\}}.$ For the single summands we have for $\nu \ge m$
\begin{align*}
	\IE &\left[ \left. \I_{\{S_{\nu} = r\}} \I_{\{S_{\nu+1} = s\}} \I_{\{t<Z_{\nu+1}\le t+h\}} \right| S_0 = s_0, Z_0 = z_0,...,S_m = r, Z_m = z, T_{m+1} > t-z \right] \\
	&= P(S_\nu = r, S_{\nu+1} = s, Z_{\nu+1} \in (t,t+h] | S_m = r, Z_m = z, T_{m+1} > t-z) \\
	&= \int\limits_{t-z}^{t-z+h} \int\limits_{t-z+t_1}^{t-z+t_1+h}\cdots\int\limits_{t-z+t_1+...+t_{\nu-m}}^{t-z+t_1+...+t_{\nu-m} + h} \sum\limits_{\substack{s_1, ..,s_{\nu - m -1}\in\IS\\s_j\ne s_{j+1}}} f_z^{rs_1}(t_1)\prod_{j=2}^{\nu - m - 1} f_{z+t_1+...+t_{j-1}}^{s_{j-1}s_j}(t_j) \cdot \\ 
	&~~~~~~~~~~~~~~~~~~~~~f_{z+t_1+...+t_{\nu-m-1}}^{s_{\nu-m-1}r}(t_{\nu-m}) f_{z+t_1+...+t_{\nu-m}}^{rs}(t_{\nu-m + 1}) dt_{\nu-m+1}\dots dt_1.
\end{align*}
The first identity applies due to the definition of semi-Markov processes. The last term only depends on $\nu-m$, but not on $m$ itself. For $\nu < m$ the conditional expectation is zero. When we denote this term as $\prescript{\nu-m}{}h_t^r(z)$ we obtain $$\IE \left[ \left.N(t+h)-N(t) \right| S_0 = s_0, Z_0 = z_0,...,S_m = r, Z_m = z, T_{m+1} > t-z \right] = \sum_{j=0}^\infty   \prescript{j}{}h_t^r(z),$$
where the series on the right-hand side meets the requirements for the function $\prescript{N}{}g_t^r(z)$ in Assertion 1. A similar argumentation yields an analoguous result for $\Delta A.$ Here, the representation $\Delta A = \sum_{\nu = 0}^\infty \int_{t\vee Z_\nu}^{(t+h)\wedge Z_\nu} \I_{\{S_\nu = r\}}\alpha(u, u - Z_\nu)$ can be applied. This completes the proof of \textit{Assertion 1}.\\
This assertion can be put to use in the following chain of equations: Let $B\in\sigma(\mathcal F_t) \cap \mathfrak P (\{Y(t) = 1\})$. Further let $\mathcal F_t(B)$ be a measurable set of the image space of $\mathcal F_t$, whose pre-image is $B$.
\begin{align*}
	\int_{\mathcal F_t(B)} & \IE \left[ \left. \Delta M \right| \mathcal F_t = f \right] dP_{\mathcal F_t}(f) \stackrel{(i)}{=} \int_{B} \IE \left[\left. \Delta M \right| \mathcal F_t \right] dP\\
	&\stackrel{(ii)}{=}  \int_{B}  \Delta M  dP \stackrel{(iii)}{=} \sum_{m=0}^\infty \int_{B \cap \{\mu(t)=m\}}  \Delta M  dP \\
	&\stackrel{(iv)}{=} \sum_{m=0}^\infty \int_{B \cap \{\mu(t)=m\}}  \IE [ \Delta M |  S_0, Z_0, \dots, S_m, Z_m,  \I_{\{Z_{m+1}>t\}}]  dP \\
	&\stackrel{(v)}{=} \sum_{m=0}^\infty \int_{B \cap \{\mu(t)=m\}} \prescript{M}{}g_t^r(Z_m) dP
	\stackrel{(vi)}{=} \sum_{m=0}^\infty \int_{B \cap \{\mu(t)=m\}} \prescript{M}{}g_t^r(t - D(t)) dP\\
	&\stackrel{(vii)}{=} \int_{B} \prescript{M}{}g_t^r(t - D(t)) dP
	\stackrel{(viii)}{=} \int_{\mathcal F_t(B)} \prescript{M}{}g_t^r(t - d_t(f)) dP_{\mathcal F_t}(f) 
\end{align*} 
Transformation theorem provides (i), definition of conditional expectations (ii), due to $B\in\sigma(\mathcal F_t)$. Then the events $\{\mu(t)=m\}$ form a partition of $\Omega$ and (iii) follows. Note that $B \cap \{\mu(t)=m\}$ is a $(S_0, Z_0, \dots, S_m, Z_m,  \I_{\{Z_{m+1}>t\}})$-measurable event, recognizable by the fact that for all $\omega\in\{\mu(t)=m\}$ the complete paths of $\mathcal F_t$ can be reconstructed. The definition of conditional expectations then gives (iv). For (v) Assertion 1 can be applied because the conditional expectation is the composition of the factorized conditional expectation and the conditioning variable. Recognize therefore that for all $\omega\in B \cap \{\mu(t)=m\}$ we have $S_m = r$ and $\I_{\{Z_{m+1}>t\}}=1\Leftrightarrow T_{m+1}>t-Z_m.$ For (vi) simply note that for all $\omega\in B\cap\{\mu(t)=m\}$ $Z_m = t-D(t)$. Finally, the summands' independence of $m$ enables (vii) and transformation theorem gives (viii).\\
Since this calculation is valid for any $B\in\sigma(\mathcal F_t) \cap \mathfrak P (\{Y(t) = 1\})$ and, hence, for all corresponding image sets $\mathcal F_t(B)$, part (a) of the lemma is proven \cite[by applying e.g.][Chapter IV, Theorem 4.4]{Elstrodt.2009}.\\
(b) First we have $|\Delta M|\le \Delta N + \Delta A$. We find that the right-hand side of (b) equals $\sup_{s\ne r, z}\prescript{N}{}g_t^s(z) + \sup_{s\ne r, z}\prescript{A}{}g_t^s(z)$. Our argumentation starts similar to that in the proof of part (a). This time we choose $B\in\sigma(\mathcal F_t) \cap \mathfrak P (\{Y(t) = 0\})$. Again, let $\mathcal F_t(B)$ be a measurable subset of the image space of $\mathcal F_t$ whose pre-image is $B$.

\begin{align*}
	\int_{\mathcal F_t(B)} & \IE \left[ \left. \Delta N \right| \mathcal F_t = f \right] dP_{\mathcal F_t}(f) \stackrel{}{=} \int_{B} \IE \left[\left. \Delta N \right| \mathcal F_t \right] dP\\
	&\stackrel{}{=}  \int_{B}  \Delta N  dP \stackrel{c.}{=} \sum_{m=0}^\infty \int_{B \cap \{\mu(t)=m\}}  \Delta N  dP \\
	&\stackrel{}{=} \sum_{m=0}^\infty \sum_{s\ne r}\int_{B \cap \{\mu(t)=m\}\cap\{S_m = s\}}  \IE [ \Delta N |  S_0, Z_0, \dots, S_m, Z_m,  \I_{\{Z_{m+1}>t\}}]  dP \\
	&\stackrel{}{=} \sum_{m=0}^\infty \sum_{s\ne r}\int_{B \cap \{\mu(t)=m\}\cap\{S_m = s\}}  \prescript{N}{}g_t^s(Z_m)  dP \\
	&\stackrel{}{\le} \sum_{m=0}^\infty \sum_{s\ne r}\int_{B \cap \{\mu(t)=m\}\cap\{S_m = s\}}  \sup_{s\ne r, z}\prescript{N}{}g_t^s(z)  dP \\
	&\stackrel{}{=} \int_{B}  \sup_{s\ne r, z} \prescript{N}{}g_t^s(z)  dP = \int_{\mathcal F_t(B)} \sup_{s\ne r, z}\prescript{N}{}g_t^s(z) dP_{\mathcal F_t}(f)
\end{align*}
Arguments for the first four lines are identical to those in (a). Then we can bound the integrand by the deterministic supremum which is furthermore independent of $m$. Then again, we combine the integration areas and apply transformation theorem. In total we thereby obtain $ \IE \left[ \left. \Delta N \right| \mathcal F_t = f \right]\le \sup_{s\ne r, z}\prescript{N}{}g_t^s(z), P_{\mathcal F_t}(f)$-a.s.. In the same manner we obtain an analogous result for $\Delta A$ und and the proof is completed. 

\end{proof}

\begin{Lem} \label{O_sq.lem}
	There is a constant $C > 0$, independent of $t$ and $\omega$, such that, for sufficiently small $h>0$, $|\IE [M(t+h) - M(t)| \mathfrak F_t]|\le C h^2$.
\end{Lem}
Before we approach the proof of Lemma \ref{O_sq.lem}, we want to point out how it enables establishing the martingale property for $M$: 
\begin{equation} \label{dM_decomposition.eq}
	\begin{split}
	\big|\IE[ M(t + h) - M(t)\big|\mathfrak F_t] \big| \le \sum_{j=1}^k \left| \IE\left[ \left.  M\left( t + \frac{jh}{k} \right) - M\left( t + \frac{(j-1)h}{k} \right) \right|  \mathfrak F_t \right] \right|  \\
	=\sum_{j=1}^k  \left| \IE \left[ \IE\left[ \left. \left. M\left( t + \frac{jh}{k} \right) - M\left( t + \frac{(j-1)h}{k} \right) \right|  \mathfrak F_{t + \frac{(j-1)h}{k}} \right] \right| \mathfrak F_t \right] \right| \\
	\le \sum_{j=1}^k C \left( \frac 1k \right)^2 = \frac Ck.
	\end{split}
\end{equation}
$\big|\IE[ M(t + h) - M(t)| \mathfrak F_t]\big|,$ itself independent of $k$, therefore is bounded by $C/k$ for arbitrary $k$ and must equal zero. So $\IE[ M(t + h) - M(t)|\mathfrak F_t]$ equals zero, too, and martingale property has been shown -- provided we can prove Lemma \ref{O_sq.lem}. \\

\begin{proof}[Proof of Lemma \ref{O_sq.lem}:]
We show four sub-statements. All of those are valid for sufficiently small $h>0$. \\
\textit{Statement 1:} There is a constant $C$, independent of $t$ and $d$, such that $$\IE [\Delta N | Z_m=t-d, S_m = r, T_{m+1}>d] = \frac{\int_d^{d+h} f_{t-d}^{rs}(u)du}{S_{t-d}^r(d)} + R_N(t, d, h),$$
where $|R_N(t,d,h)|\le C h^2$.\\ 
\textit{Statement 2:} There is a constant $C$, independent of $t$ and $d$, such that $$\IE [\Delta A | Z_m=t-d, S_m = r, T_{m+1}>d] = \frac{\int_d^{d+h} f_{t-d}^{rs}(u)du}{S_{t-d}^r(d)} + R_A(t, d, h),$$
where $|R_A(t,d,h)|\le C h^2$.\\
\textit{Statement 3:} There is a constant $C$, independent of $t$ and $d$, such that $$\IE [\Delta N | Z_m=t-d, S_m = s, T_{m+1}>d] \le Ch^2, \text{ for all } s\ne r.$$
\textit{Statement 4:} There is a constant $C$, independent of $t$ and $d$, such that $$\IE [\Delta A | Z_m=t-d, S_m = s, T_{m+1}>d] \le Ch^2, \text{ for all } s\ne r.$$
Statements 1 and 2, together with part (a) of Lemma \ref{dM_cond_F.lem}, prove the lemma on the event $\{Y(t)=1\}$. Statements 3 and 4, together with part (b) of Lemma \ref{dM_cond_F.lem}, prove it on the event $\{Y(t)=0\}$.\\
\textit{Proof of Statement 1: } First we have
\begin{equation} \label{EW_Delta_N.eq}
	\begin{split}
	\IE [\Delta N | Z_m=t-d, S_m = r, T_{m+1}>d] \\= \sum_{k=1}^\infty P(\Delta N \ge k | Z_m=t-d, S_m = r, T_{m+1}>d).
	\end{split}
\end{equation}
Due to $\{\Delta N \ge k\} \subset \{\mu(t+h) - \mu(t) \ge k \}$ we can bound the $k$th summand by $\tilde C^kh^k$ according to Lemma \ref{dMu.lem} where $\tilde C$ is independent of $t$, $d$. Therefore the series except the first summand is bounded by $\sum_{k=2}^\infty \tilde C^k h^k = (\tilde C^2 h^2)/ (1-\tilde C h)  \le \bar C h^2 $ for $h$ small. For the first summand in (\ref{EW_Delta_N.eq}) we have
\begin{align} \label{EW_DeltaN_Summand1.eq}
	P&(\Delta N \ge 1 | Z_m=t-d, S_m = r, T_{m+1}>d) \nonumber\\ &= P(Z_{m+1} \le t+h, S_{m+1} = s | Z_m=t-d, S_m = r, T_{m+1}>d) \nonumber\\
	&+ P( \text{``$\mathcal X$ jumps $\ge2$-times in $(t, t+h]$, $\ge1$-times from $r$ to $s$''} |Z_m=t-d, \nonumber\\&~~~~~~~~~~~~~~~~~~~~~~~~~~~~~~~~~~~~~~~~~~~~~~~~~~~~~~~~~~~~~~~~~~~~~~~S_m = r, T_{m+1}>d).
\end{align}
The event of the second probability, which for the sake of simplicity was only described in words here, is a subset of $\{\mu(t+h)-\mu(t)\ge2\}$ and, hence, by Lemma \ref{dMu.lem} bounded by $\tilde C h^2$ for small $h$. It remains to investigate the first summand of (\ref{EW_DeltaN_Summand1.eq}).
\begin{align*}
	P&(Z_{m+1} \le t+h, S_{m+1} = s | Z_m=t-d, S_m = r, T_{m+1}>d) \\
	&= \frac{	P(d<T_{m+1} \le d+h, S_{m+1} = s | Z_m=t-d, S_m = r)}{P(T_{m+1}>d|Z_m = t-d, S_m = r)} \\
	&= \frac{\int_d^{d+h} f_{t-d}^{rs}(u)du}{S_{t-d}^r(d)},
\end{align*}
according to the definitions of $f_{t-d}^{rs}$ and $S_{t-d}^r$. This proves \textit{Statement 1}.\\
\textit{Proof of Statemant 2: }Due to
\begin{align*}
	\Delta A = \int_t^{t+h}\alpha(u, D(u)) Y(u) du = \sum_{l=0}^\infty \int_t^{t+h} \alpha(u, u-Z_l) \I_{(Z_l, Z_{l+1}]}(u)\I_{\{S_l = r\}} du
\end{align*}
we obtain 
\begin{align} \label{EW_Delta_A.eq}
	\IE &[\Delta A | Z_m=t-d, S_m = r, T_{m+1}>d] \nonumber\\
	&= \sum_{l=0}^\infty \int_t^{t+h} \IE\big[ \alpha(u, u-Z_l) \I_{(Z_l, Z_{l+1}]}(u)\I_{\{S_l = r\}} \big| Z_m=t-d, S_m = r, T_{m+1}>d \big] du \nonumber\\
	&= \sum_{l=m}^\infty \int_t^{t+h} \IE\big[ \alpha(u, u-Z_l) \I_{(Z_l, Z_{l+1}]}(u)\I_{\{S_l = r\}} \big| Z_m=t-d, S_m = r, T_{m+1}>d \big] du, 
\end{align}
because for $l<m$ one of the indicator functions equals zero, given the condition. For $l>m$ the $l^{th}$ summand can be bounded according to (L) by
\begin{align} \label{EW_DeltaA_Absch.eq}
	\int_t^{t+h} &\IE\big[ \alpha(u, u-Z_l) \I_{(Z_l, Z_{l+1}]}(u)\I_{\{S_l = r\}} \big| Z_m=t-d, S_m = r, T_{m+1}>d \big] du \nonumber\\
	& \le C_1 \int_t^{t+h} P\big(Z_l \le u \big| Z_m=t-d, S_m = r, T_{m+1}>d \big)du \nonumber\\
	& \le C_1 \int_t^{t+h} P\big( \mu(u) - \mu(t) >l-m \big| Z_m=t-d, S_m = r, T_{m+1}>d \big)du \nonumber\\
	& \le C_1 \int_t^{t+h} \bar C_2^{l-m} h^{l-m} du \le \tilde C^{l-m+1} h^{l-m+1}
\end{align}
for $h$ small, taking advantage of Lemma \ref{dMu.lem}. Therefore the series (\ref{EW_Delta_A.eq}) from index $l=m+1$ upwards can be bounded for small $h$ by $\sum_{l=m+1}^\infty \tilde C^{l-m+1} h^{l-m+1} \le (\tilde C^2 h^2)/(1-\tilde C h)  \le \bar C h^2$.
For the first summand, i.e. $l = m$, we obtain
\begin{align*}
	\int_t^{t+h} &\IE\big[ \alpha(u, u-Z_m) \I_{(Z_m, Z_{m+1}]}(u)\I_{\{S_m = r\}} \big| Z_m=t-d, S_m = r, T_{m+1}>d \big] du\\
	&\stackrel{(i)}{=} \int_t^{t+h} \alpha(u, u-t+d) P(T_{m+1}>d-t+u| Z_m=t-d, S_m = r, T_{m+1}>d ) du \\
	&\stackrel{(ii)}{=} \int_t^{t+h} \frac{\int_t^{t+h} f_{t-d}^{rs}(u-t+d)}{S_{t-d}^r(u-t+d)} \frac{S_{t-d}^r(u-t+d)}{S_{t-d}^r(d)} du \\
	&\stackrel{(iii)}{=} \frac{\int_d^{d+h} f_{t-d}^{rs}(u)du}{S_{t-d}^r(d)}.
\end{align*}
For (i) first recognize $\I_{(Z_m, Z_{m+1}]}(u)=\I_{\{Z_m\le u\}} \I_{\{T_{m+1}> t-Z_m\}}.$ All factors that are measurable with respect to the condition are extracted from the conditional expectation whereby both of those indicators equal one. The remaining indicator is rewritten as a conditional probability. Next, remember the definition of $\alpha$, see Lemma \ref{alpha_def.leam}. Together with
\begin{align*}
	P&(T_{m+1}>d-t+u| Z_m=t-d, S_m = r, T_{m+1}>d ) \\&= \frac{P(T_{m+1}>d-t+u| Z_m=t-d, S_m = r) }{P(T_{m+1}>d | Z_m=t-d, S_m = r)},
\end{align*}
and the definition of $S_{t-d}^r$ (ii) is explained. Finally, we reduce the fraction and perform linear substitution for (iii).
Thus, Statement 2 is proven.\\
\textit{Proof of Statement 3: } For $s\ne r$ we observe, according to Lemma \ref{dMu.lem},
\begin{align*}
	P&(\Delta N \ge k | Z_m=t-d, S_m = s, T_{m+1}>d) \\
	&\le P(\mu(t+h)-\mu(t) \ge k + 1 | Z_m=t-d, S_m = s, T_{m+1}>d) \le \tilde C^{k+1} h^{k+1}, k\ge 1.
\end{align*}
This is explained by the idea that for $S_m \ne r$ at least $2k\ge k+1$ jumps have to take place in order to have $k$ jumps from $r$ to $s$. Due to (\ref{EW_Delta_N.eq}) with $s$ instead of $r$ we obtain, for small $h$,
\begin{align*}
	\IE[\Delta N |  Z_m=t-d, S_m = s, T_{m+1}>d ] \le \sum_{k=1}^\infty \tilde C ^{k+1} h^{k+1} = \frac{\tilde C^2 h^2}{1-\tilde C h} \le C h^2. 
\end{align*}
That shows \textit{Statement 3}.\\
\textit{Proof of Statement 4:} We can repeat the argumentation from the proof of Statement 2 with the difference that the summand $l=m$ in (\ref{EW_Delta_A.eq}) equals zero this time, because $\I_{\{S_m=r\}}$ is constantly zero given the condition. This suffices as proof of \textit{Statement 4} with which the proof of Lemma \ref{O_sq.lem} is completed, too. 
\end{proof}

%\newpage

\section{Uniform consistency of $\hat{\boldsymbol{\sigma}}_{\mathbf{x}}^{\mathbf{2}}$ (Theorem \ref{uniform_var.th})} \label{unifconvse}

Uniform consistency in probability of a nonparametric estimator is a rather old research topic \cite[for a survey for kernel density estimation see][]{consistenc:2010}. It was followed by almost sure uniform consistency \cite[for an example including censoring see e.g.][]{a-general-:2006}. Recall the definition of $\hat\sigma_x^2$ in Theorem \ref{satz3_1}. Furthermore we adopt from \cite{Linton.1995} the definition
\begin{align*}
	\tilde\sigma_x^2&:=\mathscr C _x^{-2} \frac 1n b^2 \sum_{i=1}^n \int K_b^2(x-X_i(s)) d\langle M_i \rangle (s)\\&=\mathscr C_x ^{-2} \frac 1n b^2 \sum_{i=1}^n \int K_b^2(x-X_i(s)) \alpha(X_i(s))Y_i(s) ds
\end{align*}
decompose
\begin{equation*}
	\sup_x|\hat\sigma_x^2 - \sigma_x^2| \le \sup_x|\hat\sigma_x^2 - \tilde \sigma_x^2| + \sup_x|\tilde \sigma_x^2 - \IE\tilde\sigma_x^2| + \sup_x|\IE \tilde \sigma_x^2 - \sigma_x^2|.
\end{equation*}
It suffices to prove convergence of all right-hand side terms. Since the proof of \cite{Linton.1995}, Theorem 2, already covers $\sup_x|\mathscr C_x - \varphi(x)|\rightarrow_P0$, it is sufficient to show convergence of $\sup_x|\mathscr C_x^2\hat\sigma_x^2 - \mathscr C_x^2\tilde \sigma_x^2|$, $\sup_x|\mathscr C_x^2\tilde \sigma_x^2 - \IE\mathscr C_x^2\tilde\sigma_x^2|$ and $\sup_x|\IE \mathscr C_x^2\tilde \sigma_x^2 - \varphi^2(x)\sigma_x^2|$. This is because from $f_n\rightarrow f$ uniformly, $g_n\rightarrow g$ uniformly, and $\inf g >0$, it follows, that, also uniformly, $f_n/g_n\rightarrow f/g$ .

For the first two of these terms, Lemma 1 of \cite{Linton.1995} shall be applied. Since point-wise convergence has been established in \cite{Linton.1995}, we restrict ourselves to the proof of (L2). Within (L2) w.l.o.g. we only consider the case $j=1$. Therefore let $x^*=(t, d^*)$. For this we have 
\begin{align} \label{L2_bound.f}
	\IE&\left[ \mathscr C_x^2\hat\sigma_x^2 - \mathscr C_x^2\tilde \sigma_x^2 - (\mathscr C_{x^*}^2\hat\sigma_{x^*}^2 - \mathscr C_{x^*}^2\tilde \sigma_{x^*}^2) \right]^2 \notag\\
	&\stackrel{(i)}{=} \IE\left[ \frac 1n b^2 \sum_{i=1}^n \int \left\{ K_b^2(x-X_i(s)) - K_b^2(x^*-X_i(s)) \right\} dM_i(s) \right]^2\notag\\
	&\stackrel{(ii)}{=}\frac1{n^2}b^4 \sum_{i=1}^n\IE \int \left\{ K_b^2(x-X_i(s)) - K_b^2(x^*-X_i(s)) \right\}^2 d\langle M_i \rangle (s)\notag\\
	&\stackrel{(iii)}{=} \frac 1n b^4 \int_{[0,1]^2} \left\{ K_b^2(x-w) - K_b^2(x^*-w) \right\}^2 \alpha(w) \varphi(w) dw\\
	&\stackrel{(vi)}{=} \frac 1n b^{-2} \int_{-x/b}^{(1-x)/b} \left\{ K^2(q) - K^2\left( q + \frac{x^*-x}{b} \right) \right\}^2 \alpha(x-bq) \varphi(x-bq) dq\notag\\
	&\stackrel{(v)}{\le} const\cdot \frac 1n b^{-4} |d^*-d|^2.\notag
\end{align}  
Definitions of $\hat\sigma_x^2$, $\tilde\sigma_x^2$ and $M_i$ provides (i). For (ii) we exploit that for general centred martingales $X$ we have $\IE X(T)^2 = \IE\langle X\rangle(T)$, that $\int HdM$ is such martingale, and that $\langle\int HdM\rangle = \int H^2 d\langle M \rangle$ for martingales $M$, which are are the difference between a counting process and its compensator \cite[see e.g.][chapter 2]{Fleming.2011}. Furthermore we independence of observations is utilized, which makes sure that we can ignore mixed terms, since these summands have expectation zero. Next, $d\langle M_i\rangle (s) = \alpha(X_i(s))Y_i(s)ds$, together with \eqref{EV.lem}, results in (iii) and the substitution $q=(x-w)/b$ provides (iv). Now consider  
\begin{align*}
	\left| K^2(q) - K^2\left( q + \frac{x^*-x}{b} \right) \right| &= \left| k^2(q_1) - k^2\left( q_1 + \frac{d^* - d}{b} \right) \right|  \prod_{j=2}^2 k^2(q_j) \\
	&\le \tilde C b^{-1} |d^*-d|  \prod_{j=2}^2 k^2(q_j), 
\end{align*} 
which holds true according to (K''), and because $x$ and $x^*$ only differ in the first coordinate. Plugging this into the equation (v) follows by recognizing that all integrands are continuous and the integration area remains compact due to the compact support of $K$.\\
Thus,  (L2) is shown due to (B') for $\mathscr C_x^2\hat\sigma_x^2 - \mathscr C_x^2\tilde \sigma_x^2$.
Further, we obtain
\begin{align*}
	\IE &\left[ \mathscr C_x^2\tilde \sigma_x^2 - \IE\mathscr C_x^2\tilde\sigma_x^2 -(\mathscr C_{x^*}^2\tilde \sigma_{x^*}^2 - \IE\mathscr C_{x^*}^2\tilde\sigma_{x^*}^2) \right]^2 = Var(\mathscr C_x^2\tilde \sigma_x^2 - \mathscr C_{x^*}^2\tilde \sigma_{x^*}^2)\\
	&\stackrel{(i)}{=}Var\left( \frac 1n b^2 \sum_{i=1}^n \int \left\{ K_b^2(x-X_i(s)) - K_b^2(x^*-X_i(s)) \right\}\alpha(X_i(s))Y_i(s)ds \right)\\
	&\stackrel{(ii)}{\le} \frac 1n b^4 \int_{[0,1]^2} \left\{ K_b^2(x-w) - K_b^2(x^*-w) \right\}^2 \alpha^2(w) \varphi(w) dw,
\end{align*}
where (i) is due to the definition of  $\tilde\sigma_x^2$ plus $d\langle M_i\rangle (s) = \alpha(X_i(s))Y_i(s)ds$, and (ii) due to Formula (10) in \cite{Linton.1995}. Because the last term is identical to (\ref{L2_bound.f}) up to the bounded factor $\alpha(w)$ inside the integral, (L2) is shown for $\mathscr C_x^2\tilde \sigma_x^2 - \IE\mathscr C_x^2\tilde\sigma_x^2$, too. 
Finally, $|\IE \mathscr C_x^2\tilde \sigma_x^2 - \varphi^2(x)\sigma_x^2|$ is to be investigated. For this we have
\begin{align*}
	|\IE \mathscr C_x^2\tilde \sigma_x^2 - \varphi^2(x)\sigma_x^2| &\stackrel{(i)}{=} \left| b^2 \int_{[0,1]^2} K_b^2(x-w)\alpha(w)\varphi(w)dw - \kappa_2^{d+1}\alpha(x)\varphi(x) \right| \\
	&\stackrel{(ii)}{=} \left| \int_{[-1,1]^2} K^2(q) \left\{ \alpha(x-bq)\varphi(x-bq) - \alpha(x)\varphi(x) \right\} dq \right|,
\end{align*}
where (i) is explained by the different definitions and \eqref{EV.lem}. (ii) is true for large $n$, hence, small $b$, due to the substitution $q = (x-w)/b$, the fact that the support of $K$ is $[-1,1]^2$, and the definition of $\kappa_2$. Because the integrand's leading factor, $K^2$, is bounded, uniform convergence on $\mathscr X$ to zero follows by Lemma \ref{uniform_int.lem}.

	\end{appendix}
\end{document}